\documentclass[journal]{IEEEtran}
\usepackage{algorithmicx}
\usepackage{algorithm,algpseudocode}

\usepackage{amsmath,graphicx}
\usepackage{float}
\usepackage{verbatim}
\usepackage{mathrsfs}
\usepackage{cite}
\usepackage{amsthm}
\usepackage{epstopdf}
\usepackage{epsfig,psfrag}
\usepackage{amsfonts,amssymb}
\usepackage{color}

\newtheorem{theorem}{Theorem}
\newtheorem{lemma}{Lemma}
\newtheorem{defn}{Definition}[section]

\usepackage{lipsum}
\usepackage{verbatim}
\usepackage{mathrsfs}
\usepackage{enumitem}
\usepackage{float}
\usepackage{tikz}
\usepackage{pgfplots}

\begin{document}
\title{Joint Detection and Localization of an Unknown Number of Sources Using Algebraic Structure of the Noise Subspace }

\author{Matthew~W.~Morency~\IEEEmembership{Student Member,~IEEE}, Sergiy~A.~Vorobyov~\IEEEmembership{Fellow,~IEEE}, and Geert Leus~\IEEEmembership{Fellow,~IEEE}
\thanks{M.~W.~Morency and G.~Leus are with the Dept. Microelectronics, School of Electrical Engineering, Mathematics, and Computer Science, Delft University of Technology, Mekelweg 4, 2628 CD Delft, The Netherlands; emails: {\tt M.W.Morency@tudelft.nl, G.Leus@tudelft.nl}. S.~A.~Vorobyov is with the Dept. Signal Processing and Acoustics, Aalto University, PO Box 13000 FI-00076 Aalto, Finland; email: {\tt svor@ieee.org}. Some preliminary results of this work have been presented at the {\it Asilomar 2016}. This work was supported in parts by the KAUST-MIT-TUD consortium grant~{OSR-2015-Sensors-2700} and Academy of Finland reserach grant No.~299243. Matthew W. Morency is supported in part by the Natural Sciences and Engineering Research Council of Canada PGS-D award.}
}

\maketitle

\begin{abstract}
Source localization and spectral estimation are among the most fundamental problems in statistical and array signal processing. Methods which rely on the orthogonality of the signal and noise subspaces, such as Pisarenko's method, MUSIC, and root-MUSIC are some of the most widely used algorithms to solve these problems. As a common feature, these methods require both a-priori knowledge of the number of sources, and an estimate of the noise subspace. Both requirements are complicating factors to the practical implementation of the algorithms, and when not satisfied exactly, can potentially lead to severe errors. In this paper, we propose a new localization criterion based on the algebraic structure of the noise subspace that is described for the first time to the best of our knowledge. Using this criterion and the relationship between the source localization problem and the problem of computing the greatest common divisor (GCD), or more practically approximate GCD, for polynomials, we propose two algorithms which adaptively learn the number of sources and estimate their locations. Simulation results show a significant improvement over root-MUSIC in challenging scenarios such as closely located sources, both in terms of detection of the number of sources and their localization over a broad and practical range of SNRs. Further, no performance sacrifice in simple scenarios is observed. 
\end{abstract}


\section{Introduction}
The problems of source localization and spectral estimation in a noisy environment are among the most fundamental problems in array processing \cite{VanTrees} and spectral analysis \cite{StoicaMoses}. Among the algorithms devised to solve these problems, subspace-based algorithms such as Pisarenko's method, MUSIC, and root-MUSIC have become ubiquitous \cite{BhaskarRao}-\cite{MahdiMe}. Subspace-based methods require two steps. First, the signal and noise subspaces must be estimated. Second, given the estimates of the signal and noise subspaces, the source locations or frequency estimates are derived with respect to some criterion, e.g. minimization of a cost function - spectral function or maximum likelihood function \cite{VanTrees}, or spectral peak-finding. The tasks of detecting the number of sources, and estimating their locations is a fundamental one in  telecommunications and radar as well as many other engineering, statistics, and scientific applications. In telecommunications, for example, it is often required to identify sources of interference and their locations in order to maintain the functioning of a system at a desired level of performance. In radar, the task of detection of the number of sources is of utmost importance. In military applications, for example, strategies to evade radar systems attempt to induce either overestimation (deployment of chaff), or underestimation (stealth) of the number of targets. Most subspace-based methods differ only in how to approach the second step. For example, MUSIC and root-MUSIC differ only in the criterion used to derive the target locations. However, it has been argued that the first step - estimating the signal and noise subspace - is far more crucial \cite{EsaVisa}. As such, substantial research efforts have been invested into providing robust estimates of signal and noise subspaces in a variety of challenging scenarios \cite{EsaVisa}-\cite{MahdiMe}. This is typically done through the eigen-decomposition of the autocorrelation matrix of a set of observations. Subspace estimation is thus reduced to a selection problem such as the root selection problem in the case of, for example, root-MUSIC \cite{EsaVisa}-\cite{MahdiMe}. However, as vector spaces, the signal and noise subspaces have a dimension which is either assumed to be known a-priori, or, perhaps more practically, must first be estimated.

While the estimate of the signal and noise subspaces is taken from the eigenvectors of the observation autocorrelation matrix, the dimension of the signal subspace is typically inferred from the distribution of the eigenvalues\cite{MahdiMe}-\cite{Zoubir}. Typically, the dimension of the signal subspace is taken to be the number of ``dominant'' eigenvalues. However, other statistical criteria such as the Akaike information criterion (AIC) and minimum description length (MDL) can be used as well drawing a parallel to model order selection \cite{KailathWax}. In recent work \cite{Zoubir}, the estimation of the number of sources has been considered as a multiple hypothesis test on the equality of eigenvalues. In order to perform the hypothesis testing, multiple instances of the observation autocorrelation matrix must be generated, and multiple eigen-decompositions performed. The assumption which underlies all of these methods, though, is that information about the number of sources is contained in the eigenvalues, while the eigenvectors themselves are ignored.

In this paper, we argue that under certain assumptions, the algebraic structure of the eigenvectors themselves contains a great deal of information about the number of sources, as well as their locations. Specifically, the noise eigenvectors are argued to lie in a univariate polynomial ideal generated by a single element in the univariate polynomial ring. The degree of the generator thus corresponds to the number of targets and the roots of the generator correspond to their locations.

To exploit the above described algebraic structure of the noise subspace we propose two algorithms, the first of which does not need to estimate the noise subspace (however, it can use any estimate of the noise subspace), and the second of which uses the structure of the noise eigenvectors to provide the subspace estimate. The first algorithm uses hierarchical clustering to first locate clusters of roots which are tightly located. An estimate of the target location is then based on the phase information of all the roots in the cluster. The second method uses the method of Lagrange multipliers proposed in \cite{KarmarLakshman} to produce an estimate of the greatest common divisor of the noise eigenvectors (which are viewed as polynomials) in an optimal way. This method proceeds in two steps: an initial estimate of the source locations (and number of sources), followed by a root ``refining'' Gauss-Newton iteration. The root-refinement produces a certificate of a necessary condition of optimality of the estimate of the approximate greatest common devisor (GCD), and can be used with any initial estimate, including that of the proposed root-clustering algorithm.

This paper contains several contributions.\footnote{Some preliminary results have been reported in the conference publication \cite{Asilomar17}.}
\begin{itemize}
	\item The first is a thorough analysis of the noise subspace of the observed signal covariance matrix, and a description of its algebraic structure. 
	
	\item Based on this algebraic structure, we formulate a new noise subspace selection criterion that allows for accurate noise subspace estimation by simply computing the GCD between two randomly selected eigenvectors of the signal covariance matrix.
	
	\item We show that in practice for small sample size and other imperfections, the source localization problem is equivalent to the approximate GCD problem, and thus amenable to solution via techniques for finding approximate GCDs. 
	
	\item We prove that the formation of the root-MUSIC polynomial is optimal with respect to a measure of the perturbation of the observed noise eigenvectors. To our knowledge, this result has not been reported previously in the signal processing literature. 
	
	\item We provide a certificate to verify the optimality of any subspace based localization algorithm. 
	
	\item We develop practical algorithms for joint detection and localization of an unknown number of sources using the algebraic structure of the noise subspace, which are based on the approximate GCD calculation. 
	
	\item Via simulations, we show that the proposed algorithms sacrifice no performance in simple scenarios, while providing a tangible benefit in challenging scenarios such as closely located sources in moderate signal-to-noise ratio (SNR) conditions.
\end{itemize}

{\it Notation:} Throughout this paper bold upper-case letters denote matrices, bold lower-case letters stand for vectors, upper-case letters are constants, and lower-case letters are variables. The $N \times N$ identity matrix is denoted as $\mathbf I_N$, while $\boldsymbol{0}$ stands for the vector of zeros and $\boldsymbol{0}_{N \times M}$ is the matrix of zeros of size $N \times M$. The complex Gaussian distribution of a random vector with zero mean and covariance matrix $\mathbf C$ is denoted as $\mathcal{CN} (\boldsymbol{0}, \mathbf C)$. Real and imaginary parts of a complex number are denoted as ${\rm Re} (\cdot )$ and ${\rm Im} (\cdot )$, respectively, and $j \triangleq \sqrt{-1}$. The notation $\rm{range}(\cdot)$ is used for the operator that returns the column space of its matrix argument, while $(\cdot)^H$, $(\cdot)^T$, $(\cdot)^*$, $(\cdot)^{\dagger}$, $\| \cdot \|$, and ${\rm deg} (\cdot)$ denote respectively the Hermitian transpose, transpose, conjugate of a complex number, Moore-Penrose pseudo-inverse, Euclidian norm of a vector, and degree of a polynomial. In addition, $\mathbb{K}$ stands for a base field and $\mathbb{K}[x_1,\cdots,x_n]$ denotes the polynomial ring with coefficients in $\mathbb{K}$ in $n$-variates. A {\it commutative ring} is defined as a set that is closed under two different operations, namely addition and subtraction. In addition to being commutative, commutative rings are associative, left and right distributive, and contain identity elements for both the addition and multiplication operation (denoted as 0 and 1).

The paper is organized as follows. The data model and basics of subspace-based localization algorithms are given in Section \ref{sec:DataProblem} together with the problem description. The algebraic structure of the noise subspace is derived in Section~\ref{sec:SelecCriterion}. In this section, the new noise subspace selection criterion is also formulated and the main theoretical results for the paper relating the localization problem of an unknown number of sources to the approximate GCD calculation are also given. We propose algorithms which leverage the algebraic structure of the noise subspace and address the corresponding approximate GCD problem in Section~\ref{sec:Algorithms}. Simulation results follow in Section~\ref{Simul}. The paper is concluded in Section~\ref{Concl}, and some technical derivations are given in Appendix.

\section{Data Model, Subspace-Based Methods and the Problem}
\label{sec:DataProblem}
\subsection{Data Model}
\label{subsec:data}
Consider $L$ independent narrow-band Gaussian sources in the far-field impinging upon a uniform linear array (ULA)\footnote{We will show in Section~\ref{Non-uniform} how the ULA assumption can be relaxed for our developments in the paper, but for convenience we start here with ULA.} of $N$ antenna elements with inter-element spacing $\lambda_c/2$, where $\lambda_c$ is the carrier wavelength. The signal observed at the antenna array at time $t$ can be written as
\begin{align} 
\mathbf x(t) &= \mathbf A \mathbf s(t) + \mathbf n(t)  \label{eq:obsignal}
\end{align}
where $\mathbf A \triangleq [\mathbf a(\theta_1),\cdots,\mathbf a(\theta_L)]$, $\mathbf [\mathbf a(\theta)]_n \triangleq \alpha^{n-1}$, $\alpha \triangleq e^{j \pi sin(\theta_l)}$, $\theta_l$ is the DOA of the $l$-th target, $\mathbf s(t)$ is an $L \times 1$ vector of Gaussian i.i.d. equal power source signals, which can be assumed deterministic or stochastic Gaussian distributed noise vector with zero mean and covariance $\mathbf S$, i.e., $\mathcal{N}(\boldsymbol{0}, \sigma_s^2 \mathbf I_L)$, $\sigma_s^2$ is the source power, $\mathbf n(t) \sim \mathcal{N} (\boldsymbol{0}, \sigma_n^2 \mathbf I_N)$, is the $N \times 1$ Gaussian distributed with zero mean and covariance $\sigma_n^2 \mathbf I_N$ sensor noise vector, $\mathbf I_N$ denotes the identity matrix of size $N \times N$, $\sigma_n^2$ is the noise power, and $\boldsymbol{0}$ stands for the vector of zeros. Collecting $T$ observations $\mathbf x(t)$, the sample covariance matrix (SCM) can be computed as 
\begin{align} 
\hat{\mathbf R}_{xx} &\triangleq \frac{1}{T}\sum_{t=1}^T \mathbf x(t) \mathbf x^H(t) \nonumber \\
							 &\approx \mathbf A \mathbf A^H \sigma_s^2 \mathbf I_L+ \sigma_n^2 \mathbf I_N \label{eq:cov_matrix}
\end{align}
where the approximate equality holds for a sufficiently large sample size $T$ and follows after substituting \eqref{eq:obsignal} to the first row of \eqref{eq:cov_matrix}. 

For signal model \eqref{eq:obsignal}, $\hat{\mathbf R}_{xx}$ has full rank almost surely if $T \geq N$. Since the true covariance matrix $\mathbf R_{xx}$ is Hermitian by definition, it has a full set of real eigenvalues, and an eigenbasis, allowing us to write the eigenvalue decomposition for the estimate $\hat{\mathbf R}_{xx}$ as
\begin{align}
\hat{\mathbf R}_{xx} &= \mathbf Q \boldsymbol \Lambda \mathbf Q^H \nonumber \\
&= \mathbf Q_s \boldsymbol{\Lambda}_s \mathbf Q_s^H + \mathbf Q_n \boldsymbol{\Lambda}_n \mathbf Q_n^H \label{eq:eig_decomp}
\end{align}
where $\mathbf Q_s$, $\boldsymbol{\Lambda}_s$, and $\mathbf Q_n$, $\boldsymbol{\Lambda}_n$ are the matrices of eigenvectors and eigenvalues for the signal and noise subspaces, respectively, with $\boldsymbol{\Lambda}_s$ and $\boldsymbol{\Lambda}_n$ being diagonal matrices. 

A major difficulty with subspace based source localization methods introduced next  is that we only have access to the matrix $\mathbf Q$ and must choose which columns belong to $\mathbf Q_s$ and which belong to $\mathbf Q_n$. In other words, the dimension of the signal subspace is not known in general.

\subsection{Subspace Based Methods}
Using the property of the decomposition \eqref{eq:eig_decomp}, we have $\mathbf Q_s \perp \mathbf Q_n$, and as $T \to \infty$, $L$ can be exactly estimated and $\rm{ran}(\mathbf Q_s) = \rm{ran}(\mathbf A)$, which implies that $\mathbf Q_n^H \mathbf a(\theta_l) = \boldsymbol{0}$, $\forall l$, where $\theta_l$ is the $l$-th source direction of arrival (DOA). It is this property which is exploited by subspace based methods such as Pisarenko's method, MUSIC, and root-MUSIC. 

As was mentioned in the introduction, MUSIC and root-MUSIC differ only in how the DOAs are retrieved from the subspace estimates. Given an estimate of the number of sources $\hat{L}$ and a corresponding estimate of the noise subspace $\mathbf Q_n$, both algorithms derive their estimates from the function $J(\theta) = \mathbf a^H(\theta) \mathbf Q_n \mathbf Q^H_n \mathbf a(\theta)$. The MUSIC algorithm takes the $\hat{L}$ largest peaks of the spectral function $J^{-1}(\theta)$ as $\theta$ is varied across the entire angular sector.
The root-MUSIC algorithm for ULAs first treats each column of $\mathbf Q_n$ as the vector of coefficients of a polynomial in the field defined by this vector, i.e., $\mathbb{C}[x]$, and factors each polynomial using a root-finding algorithm. Then, the $\hat{L}$ roots which are closest to the unit circle are taken to be the roots which correspond to the targets. Specifically, the $\hat{L}$ roots are found as the points $z \in \mathbb{C}$ minimizing 
\begin{align}
 \mathbf d^H(z) \mathbf Q_n \mathbf Q^H_n \mathbf d(z) \label{eq:rMUSIC}
\end{align}
where $ [\mathbf d(z)]_n = z^n,\ z \in \mathbb{C}$. The difference between the definitions of $\mathbf a(\theta)$ and $\mathbf d(z)$ are the domains. The domain of $\mathbf d(z)$ is the field of complex numbers, whereas $\mathbf a(\theta)$ accepts arguments only from the range $[0,2\pi)$ corresponding to the unit circle. This is the reasoning behind the root selection criterion of root-MUSIC.

\subsection{Problem Description}
\label{ProblemForm}
Clearly, knowledge of the number of sources is required in the algorithms discussed above. This knowledge is typically based on the distribution of the eigenvalues of $\hat{\mathbf R}_{xx}$. Information criteria such as the Akaike information criterion (AIC) or minimum description length (MDL) can, for instance, be used to identify the number of sources in drawing a parallel to model order selection \cite{KailathWax}. MDL was shown to be a consistent estimator, and returns the correct number of targets in the infinite sample regime, while AIC is not even consistent \cite{Rissanen}. Indeed, AIC is well known to overestimate the model order. It was demonstrated in \cite{XuPierreKaveh}, however, that the methods based on the distribution of the eigenvalues of $\hat{\mathbf R}_{xx}$, such as AIC, and MDL, quickly break down in non-ideal scenarios, such as model mismatch or closely located sources. In non-ideal scenarios, all aforementioned methods can lead to model order under- and overestimation even for infinite sample size. We demonstrate this through the following example. 
 
{\it Illustrative Example 1:} Fig.~\ref{fig:eigenvalues} shows the eigenvalue distribution of a single instance of $\hat{\mathbf R}_{xx}$ corresponding to a ULA of $N = 10$ elements, at $10$~dB signal-to-noise ratio (SNR), with $T=100$ snapshots, and two closely located sources impinging from $31^o$ and $32^o$. The difference between the second eigenvalue and the last eigenvalue is $0.1396$, while the first eigenvalue is over $100$ times larger than the second eigenvalue. Thus, in such non-ideal scenario of closely located sources, it would be difficult to conclude that there is more than one source on the basis of the eigenvalues alone. However, AIC estimates the presence of between $2$ and $5$ sources, while MDL estimates the presence of between $1$ and $4$ sources. We propose a method to estimate the dimension of the signal subspace which, in practice, never over-estimates the number of sources even in such a challenging scenario. Under-estimation is sometimes unavoidable due to source merging.

\begin{table*}[t]
	\centering
	\caption{Eigenvalues of the SCM corresponding to illustrative example $1$.}
	\begin{tabular}{||c c c c c c c c c c||} 
		\hline
		$\lambda_1$ & $\lambda_2$ & $\lambda_3$ & $\lambda_4$ & $\lambda_5$ & $\lambda_6$ & $\lambda_7$ & $\lambda_8$ & $\lambda_9$ & $\lambda_{10}$ \\
		\hline\hline
		20.3978 & 0.2015 & 0.1499 & 0.1289 & 0.1179 & 0.1148 & 0.0919 & 0.0859 & 0.0675 & 0.0619 \\
		\hline
	\end{tabular}
\end{table*}
 
 \begin{figure}[t]
 	\begin{center}
 		\includegraphics[width=\linewidth,trim=0 250 0 200]{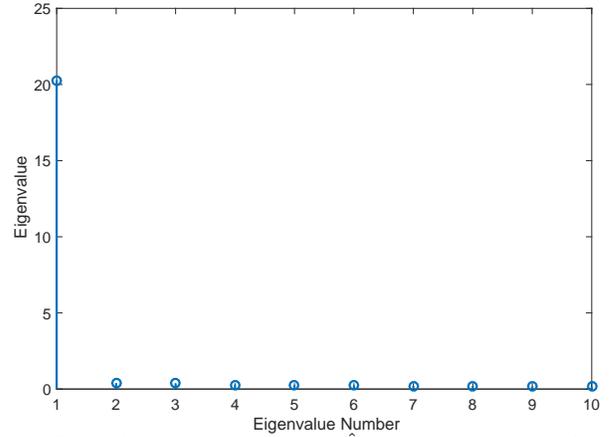}
 	\end{center}
 	\vspace{2ex}
 	\caption{Eigenvalues for one instance of $\hat{\mathbf R}_{xx}$. ULA of $N = 10$ elements, $T = 100$ snapshots, $10$~dB SNR, two sources located at $31^o$ and $32^o$.} \label{fig:eigenvalues}
 \end{figure}
 
The phenomenon depicted in Fig.~\ref{fig:eigenvalues} is referred to as the ``breakdown'' of the detector \cite{XuPierreKaveh}. Thus, it can be concluded that the distribution of the eigenvalues alone may be unreliable for the model order estimation especially when the angular separation is small and/or the sample size is small \cite{MahdiMe}, but even in the case of infinite sample size under some non-idealities in the model. As a way around of the aforementioned ``breakdown'' phenomenon, the eigenvectors themselves will be considered in this paper. Interestingly, the authors of \cite{XuPierreKaveh} presented an alternative to using only eigenvalue information, which also considers the eigenvectors themselves, and develop a statistic based on the average projection (taken over all angles) of the array steering vector on the null-space defined by the eigenvectors currently being tested. Moreover, a similar test was presented in \cite{VanTrees}, where the source DOAs are assumed to be known approximately. Signal eigenvectors are then chosen to be the eigenvectors onto which the presumed steering vector has drastically larger projections. In contrast with these methods, our methods in this paper assume no knowledge and make no use of the steering vector and its projections. Instead, our methods make use of the algebraic structure present in the noise eigenvectors, which we find out and explain for the first time, under the data model in Subsection~\ref{subsec:data}.

More generally, it is known that even if the root-MUSIC algorithm is provided with the correct number of sources, in challenging scenarios such as the one depicted in Fig.~\ref{fig:eigenvalues} or, if the sample size is small, the performance of the algorithm is generally quite poor. This is because of a deficiency in the classic root-selection criterion, which as we argue in the next section should be replaced by a new criterion based on the algebraic structure of the noise subspace. It will be then shown that methods developed based on the new criterion overcome the aforementioned problem and lead to better performance in challenging scenarios.

\section{Algebraic Structure of the Noise Subspace} \label{sec:SelecCriterion}
Let us first assume that the number of sources $L$ is known. Then for known $L$ we aim at explaining the algebraic structure of the noise subspace for the data model \eqref{eq:obsignal}. Pursuant to the discussions in the previous section, the assumption of known $L$ will later be dropped, as we introduce methods to estimate the number of sources. 

As was described in the previous section, subspace based methods leverage the orthogonality of the noise subspace estimate, i.e., $\mathbf Q_n$, and the steering vector corresponding to the DOAs of the sources. Specifically, when $T \to \infty$, the condition that
\begin{align}
\mathbf A^H \mathbf Q_n &= \mathbf 0_{L,N-L} \label{eq:MUSIC_cond}
\end{align}
must hold. In other words, the column space of $\mathbf Q_n$ is entirely contained in the null-space of $\mathbf A^H$. 

To explain the algebraic structure of the noise subspace $\mathbf Q_n$, we will need \eqref{eq:MUSIC_cond} and will also make use of some basic concepts from algebraic geometry. Since algebraic geometry is not yet a common tool in signal processing we introduce the concepts that we need to aid in the understanding of the following contents. Algebraic geometry is concerned with the relations between sets of polynomials called ideals (algebraic objects) and their associated zero loci called varieties (geometric objects) \cite{IVA}.

\begin{defn}
	An ideal $\mathcal{I}$ in $\mathbb{K}[x_1,\cdots,x_n]$ is a subgroup of $\mathbb{K}[x_1,\cdots,x_n]$ with the property that $\forall a \in \mathcal I,\ r \in \mathbb{K}[x_1,\cdots,x_n], a \cdot r \in \mathcal{I}$.
\end{defn}

As an example, take the commutative ring of univariate polynomials over $\mathbb C$, written $\mathbb{C}[x]$, and as a subset, take the set of all polynomials with a common root at $\alpha \in \mathbb{C}$\cite{Vinberg}.

\begin{defn}
	An algebraic variety is a subset of $\mathbb{K}^N$ such that $V(\mathcal I) \triangleq \{\mathbf{p} \in \mathbb{K}^N | f(\mathbf{p}) = 0,\ \forall f \in \mathcal{I} \subset \mathbb{K}[x_1,\cdots,x_n]\}$.	
\end{defn}

Thus, an algebraic variety is described by the polynomials vanishing on it. One can similarly describe an ideal by the set on which every member vanishes\cite{IVA}.

\begin{defn}
	A polynomial ideal given a variety $V \in \mathbb{K}^N$ is a set of polynomials with the property that $\mathcal{I}(V) = \{f \in \mathbb K [x_1,\cdots,x_N] | f(\mathbf p) = 0,\ \forall \mathbf{p}\in V\}$. \label{def:definition3}
\end{defn}

For example, take the variety consisting of two points on the real line $V = \{1,2\}$. The ideal corresponding to this variety is the set of polynomials with at least one root at each of $x = 1$ and $x = 2$.

Ideals are generated by elements contained within them, much the same way that a vector space is spanned by linearly independent vectors contained in the space. An ideal which is generated by a single element is a principal ideal\cite{IVA}.
\begin{defn}
	The principal polynomial ideal generated by $f$, denoted as $\langle f \rangle$, is the set \\ $\{h\ | \ h = f \cdot g,\ g \in \mathbb{K}[x_1,\cdots,x_N] \}$. \label{def:principal}
\end{defn}
All univariate polynomial ideals are principal and can thus be generated by a single element of $\mathbb{K}[x]$ which is the greatest commmon divisor (GCD) of all the polynomials in the ideal $\mathcal{I}$. 

Now we observe that $\mathbf d^H(z)\mathbf q_i$ is simply the evaluation of a polynomial with coefficients defined by the entries of the column $\mathbf q_i$ of $\mathbf Q$ at $z$. Since each column of the noise subspace $\mathbf Q_n$ must lie in the null-space of $\mathbf A^H$ accordign to \eqref{eq:MUSIC_cond}, we can use Definition~\ref{def:definition3} to assert that the columns of $\mathbf Q_n$ must lie in a univariate polynomial ideal $\mathcal{I}(V)$ where $V = \{\alpha_1,\cdots,\alpha_n\}$ are the generators of the columns of $\mathbf A$. Pursuant to Definition~\ref{def:principal}, this implies that the columns of $\mathbf Q_n$ be generated by a single element $Q(x) \in \mathbb{C}[x]$. This generator being, specifically,

\begin{align}
Q(x) = \prod_{l=1}^L (x - \alpha_l). \label{eq:generator}
\end{align}
where $\alpha_l$ are complex numbers whose phase arguments correspond to the target locations.

This extends trivially to the root-MUSIC polynomial $\mathbf d^H(z) \mathbf Q_n \mathbf Q^H_n \mathbf d(z)$ as it is merely the $2$-norm of the inner product $\mathbf d^H(z) \mathbf Q_n$. Specifically, the polynomial ideal which describes the noise subspace is a function of the target locations, parametrized by $\alpha_l$, $\forall l$. This yields the following alternative noise subspace selection criterion.

Based on the algebraic structure of the noise subspace, the new noise subspace selection criterion can be formulated as follows.

{\bf Noise Subspace Selection Criterion:} {\it Given $N$ polynomials whose coefficients are described by the $N$ eigenvectors of $\hat{\mathbf R}_{xx}$, select the largest subset which lies in a polynomial ideal. The cardinality of this set is $N - \hat{L}$, where $\hat{L}$ is the degree of the generator of the polynomial ideal. Its factors are injectively related to the source locations.

This criterion describes the essential difference between our approach, and the dominant approach of source estimation based on eigenvalues, and estimation based on eigenvectors. We aim to estimate a single object, the maximal degree GCD, of a subset of the observed eigenvectors. The properties of this object are the solutions to the separate problems. Specifically, the maximal degree is the number of targets, and the factors themselves are injectively related to the target locations. Notably, this implies that for every degree of this GCD, there is a factor which provides the estimate of the target location. In illustrative example 1, however, it was shown that this is not necessarily true of the eigenvalue-detection eigenvector-estimation paradigm. Targets could be estimable based on the observed eigenvectors, but not detectable on the basis of the eigenvalues alone.}

It is worth mentioning here an intresting connection to the theory of sums of squares (SOS) polynomials. The matrix $\mathbf Q_n \mathbf Q^H_n$ is positive semidefinite (PSD), and thus, $\mathbf d^H(z) \mathbf Q_n \mathbf Q^H_n \mathbf d(z)$ is a globally non-negative function over $\mathbb{C}$. Hilbert proved that a univariate polynomial is globally non-negative if and only if it is a SOS polynomial \cite{Blekherman}. That is, $J(\theta) = 0$, or equaivalently the spectral function has a peak, if and only if the columns of $\mathbf Q_n$ lie in a univariate polynomial ideal, which gives the above noise subspace selection criterion. 

\subsection{Exact $\hat{\mathbf R}_{xx}$}
\label{sec:algorithms}
In the previous sections, the analysis assumed an infinite sample size, in which case the estimate $\hat{\mathbf R}_{xx}$ is almost exact, that is, $\hat{\mathbf R}_{xx} \approx \mathbf R_{xx} = \sigma_s^2\mathbf A \mathbf A^H + \sigma_n^2 \mathbf I$. In this case, the noise subspace lies exactly in a polynomial ideal. Moreover, consistent eigenvalue methods for estimating the number of targets allows the error free separation of the matrix $\mathbf Q$ into signal and noise component matrices.
 
As discussed in Subsection~\ref{ProblemForm}, the estimation of the noise subspace based on the conventional criteria may be inaccurate if there exist data model non-idealities. The new subspace selection criterion allows, in contrast, for accurate noise subspace estimation in terms of simply computing the greatest common divisor (GCD) between two randomly selected columns of $\mathbf Q$, where Euclid's algorithm can be used for computing GCD \cite{KarmarLakshman}. Indeed, if two noise eigenvectors are selected, then the GCD is a polynomial, and its factors, $\alpha_l$, are the complex generators corresponding to the source locations. Two things are worth noting about this selection algorithm. The first is that the maximum degree GCD of the noise eigenvectors cannot exceed $N-L$, where $L$ is the true dimension of the signal subspace. The second is that the signal eigenvectors which are linear combinations of polynomials themselves almost surely do not share roots. These two statements are proven in the following two lemmas.

\begin{lemma}
\label{lem:lemma1}
The degree of the largest degree GCD of $\mathbf q_{L+1},\cdots,\mathbf q_{N}$, where $\mathbf q_{L+1},\cdots,\mathbf q_{N}$ are the noise eigenvectors of $\mathbf R_{xx}$ is upper-bounded by $L$.	
\end{lemma}

\begin{proof}
In section \ref{sec:SelecCriterion} it was shown that the noise eigenvectors of $\mathbf R_{xx}$ when interpreted as the coefficients of a univariate polynomial exist in a univariate polynomial ideal generated by $Q(x) = \prod_{l=1}^{L} (x - \alpha_l)$. Let us denote the restriction of this polynomial ideal to polynomials of degree less than or equal to $N - 1$ as $\langle Q(x) \rangle |_{N-1}$. Using Vi\`ete's formulas we can produce a vector of coefficients of $Q(x)$ denoted as $\mathbf q$, which allows us to write a basis for the noise subspace in Toeplitz form as

\begin{align}
\mathbf B \triangleq \begin{bmatrix}
\mathbf q & 0 & \cdots & 0 \\
\boldsymbol 0 & \mathbf q & \vdots & \boldsymbol{0} \\
\vdots & \boldsymbol{0} & \ddots & \mathbf q 
\end{bmatrix}
\end{align}  
The noise eigenvector matrix can then be written as $\mathbf Q_n = \mathbf B \mathbf R$ where $\mathbf R$ is some invertible matrix. Assume that the dimension of the signal subspace is $L$, and that the noise eigenvectors have a maximal degree GCD of degree $L+1$. The dimension of $\mathbf{B}$ must then necessarily be $N \times (N-L-1)$ since the length of $\mathbf q$ must be $L + 2$. However this is a contradiction of the assumption that $\mathbf R_{xx}$ was full rank and Hermitian, and thus has a complete orthogonal eigenbasis, which shows the claim.
\end{proof}

We note here that $\mathbf B$ is known as a convolution matrix having the property that an $m$ degree polynomial $f(x) = q(x)u(x)$ can be represented as a vector in $\mathbb{C}^{m+1}$ as $\mathbf f = \mathbf B \mathbf u$ where $\mathbf q$ and $\mathbf u$ are the vector representations of $q(x)$ and $u(x)$.

\begin{lemma}
The signal eigenvectors of $\mathbf R_{xx}$, $\mathbf q_{1},\cdots,\mathbf q_{L}$ almost surely have a constant GCD.
\end{lemma}

\begin{proof}
From \eqref{eq:obsignal} we can write $\mathbf Q_s$ as $\mathbf A \mathbf R$ where $\mathbf R$ is some invertible matrix. Then the columns of $\mathbf Q_s$ share a root if and only if
\begin{align}
\mathbf Q_s^T \mathbf z &= \mathbf R^T \mathbf A^T \begin{bmatrix}
1 \\
z \\
\vdots \\
z^N-1
\end{bmatrix} = \boldsymbol{0}
\end{align}
for some $z \in \mathbb{C}$. Since $\mathbf R$ was invertible, the above equation has a solution only if $\mathbf A^T \mathbf z = 0$. That is the polynomials $f_l(z) = 1 + \alpha_l z + \alpha_l^2 z^2 + \cdots + \alpha^{N-1} z^{N-1} = 0$ simultaneously. Considering $\alpha$ to be a variable (since we do not know the $\alpha_l$ in question) the polynomial $f_l(z)$ becomes a multivariate polynomial $f(\alpha,z)$ whose zero locus is a set of dimension $1$ in the affine space $\mathbb{C}^2$. This alone suffices to show the almost sure claim made in the lemma.
\end{proof}

It is worthwhile to consider whether the zero locus of $f(\alpha,z)$ has multiple solutions $\alpha$ on the unit circle for a given $z$ at all, or in how many places. However, this is unnecessary for the preceding claims and we thus leave it as a conjecture.

\subsection{$\epsilon$-Ideal and Approximate GCD} 
In practice, the sample size is finite and typically small. As a result, the noise eigenvectors are perturbed from being exactly in a univariate ideal, which makes the application of the new noise subspace selection criterion and GCD-based algorithm explained in the previous subsection not straightforward. Indeed, the perturbed noise eigenvectors lie in an $\epsilon$-ideal \cite{BorderBases}. 

The $\epsilon$-ideal structure of $\mathbf Q_n$ implies that each noise eigenvector has, as a factor, the perturbed generator
\begin{align}
Q(x) = \prod_{l=1}^L (x - \alpha_l + \epsilon_l) \label{eq:epsilon}
\end{align}
where $\epsilon_l$ are small random perturbations \cite{BhaskarRao}. Then, the problem of finding an approximate GCD can be seen as finding $\hat{\alpha_l}$ such that 
\begin{align}
\hat{Q}(x) = \prod_{l=1}^L (x-\hat{\alpha_l})
\end{align}
is as ``close'' to \eqref{eq:epsilon} as possible, in some appropriate sense. What we know are the coefficients of \eqref{eq:epsilon} as some of the columns of the matrix $\mathbf Q$. There are two competing notions of ``closeness'' for approximate GCD's in computer algebra. The first is based on the perturbations of the coefficients themselves. The second is based on the roots of the polynomials. The algorithms we present in this paper use both. We investigate the consequences of the two definitions and their corresponding algorithms in the simulation section.

In the following derivation, we investigate ``nearness'' with respect to the polynomial coefficients. Polynomials are taken from $\mathbb{C}[z]$ with complex valued arguments and coefficients. Consider two polynomials $f(z) = \sum_{i=0}^{N} f_i z^i$ and $g(z) = \sum_{i=0}^{N} g_i z^i$ which do not have a non-trivial GCD, but have an approximate GCD of the form \eqref{eq:epsilon}. Given two polynomials $f(z)$ and $g(z)$ which ``almost'' share a root the goal is to perturb the coefficients of the polynomials $f(z)$ and $g(z)$ to polynomials $\hat{f}(z)$ and $\hat{g}(z)$ such that
\begin{align}
\hat{f}(\alpha) &= \hat{g}(\alpha) = 0
\end{align}
for some $\alpha \in \mathbb{C}$, where $\hat{f}_i = f_i + \lambda_i$ and $\hat{g}_i = g_i + \mu_i,\ \lambda_i,\ \mu_i\ \in \mathbb{C}$. Then $x - \alpha$ would be the GCD of $\hat{f}(z)$ and $\hat{g}(z)$. The GCD that corresponds to the minimal perturbation of the coefficients of $f(z)$ and $g(z)$ will then be known as the ``nearest GCD.'' In this sense, we will show that the formation of the root-MUSIC polynomial is optimal. That is, in order to solve the nearest GCD problem, the root-MUSIC polynomial results as the minimizer of a Lagrange multiplier problem involving the coefficient perturbations. The proof contained in this section first appeared in \cite{KarmarLakshman}. We reproduce it here for clarity and note that the minimizer of the Lagrange multiplier problem is, in fact, the root-MUSIC polynomial \eqref{eq:rMUSIC}.  

The problem of finding the nearest GCD defined above was first formulated in \cite{KarmarLakshman} as
\begin{align}
\min_{\boldsymbol{\mu,\lambda}}\ \boldsymbol{\epsilon} \quad \rm{s.t.}\ \hat{\mathit{f}}(\alpha; \boldsymbol{\mu},\boldsymbol{\lambda}) = \hat{\mathit{g}}(\alpha; \boldsymbol{\mu},\boldsymbol{\lambda}) = 0 \label{eq:poly_opt}
\end{align}
where $\boldsymbol{\epsilon} = \sum_{i=0}^{N} \lambda_i \lambda_i^* + \mu_i \mu_i^*$. Then the following theorem holds. The following derivation is an exposition of the proof in \cite{KarmarLakshman} and connects the theory of nearest GCDs to the solution of problem \eqref{eq:rMUSIC}.

\begin{theorem}
	The minimizer of $\boldsymbol \epsilon$ for $f_1,\cdots,f_k$ where $f_1,\cdots,f_k$ are polynomials whose coefficients are defined by the noise eigenvectors $u_1,\cdots,u_k$ is the root-MUSIC polynomial \eqref{eq:rMUSIC}.	
\end{theorem}
\begin{proof}
	Details of the following derivation are given in the appendix. Introducing the multipliers $a,b,c,d$, we write the Lagrangian of problem \eqref{eq:poly_opt} as
	\begin{align} \label{Lagrangian}
	\mathcal{L}(\lambda,\mu,a,b,c,d) &= \boldsymbol{\Sigma} + 2a\mathrm{Re}(\hat{f}(\alpha)) + 2b\mathrm{Im}(\hat{f}(\alpha)) +\nonumber\\ &2c\mathrm{Re}(\hat{g}(\alpha)) + 2d\mathrm{Im}(\hat{g}(\alpha))
	\end{align}
	At optimality with respect to $\boldsymbol{\lambda}$ and $\boldsymbol{\mu}$, the real and imaginary parts of both functions must be 0. Differentiating \eqref{Lagrangian} with respect to the real and imaginary parts of $\lambda_i$ and $\mu_i$ and solving for the multipliers by substituting back into the constraint of \eqref{eq:poly_opt}, it can be found that
	\begin{align}
	\lambda_i = -(a+jb)\alpha^{*i},\ &\quad \mu_i = -(c+jd)\alpha^{*i} \label{eq:lagrange1} \\
	a + jb = -\frac{f(\alpha)}{\sum_{k=0}^{N-1}(\alpha^* \alpha)^k},\    
	&\quad c + jd = -\frac{g(\alpha)}{\sum_{k=0}^{N-1}(\alpha^* \alpha)^k} \label{eq:lagrange2}.
	\end{align}
	
	By substituting \eqref{eq:lagrange2} into \eqref{eq:lagrange1}, the expression for $\boldsymbol \epsilon$ at its minimum can then be written entirely in terms of the original polynomials $f(z)$ and $g(z)$ and $\alpha$ as
	\begin{align}
	\boldsymbol{\epsilon}_{\rm min} &= \frac{\sum_{i=0}^{N-1}[f(\alpha)f^*(\alpha)\alpha^i \alpha^{*i} + g(\alpha)g^*(\alpha)\alpha^i \alpha^{*i}]}{(\sum_{k=0}^{N-1}(\alpha \alpha^*)^k(\sum_{k=0}^{N-1}(\alpha \alpha^*)^k)^*} \nonumber \\
	&= \frac{f(\alpha)f^*(\alpha) + g(\alpha)g^*(\alpha)}{\sum_{k=0}^{N-1}(\alpha \alpha^*)^k}. \label{eq:sigma_M}
	\end{align}
	
	The derivation in \cite{KarmarLakshman} of \eqref{eq:sigma_M} can easily be extended to an arbitrary number of polynomials, yielding the following formula for $\boldsymbol{\epsilon}_{\rm min}$
	\begin{align}
	\boldsymbol{\epsilon}_{\rm min} &= \frac{\sum_{l=1}^{L} f_l(\alpha) f_l^*(\alpha)}{\sum_{k=0}^{N-1} (\alpha \alpha^*)^k} \label{eq:karmar_music}
	\end{align}
	
	Noting that the noise eigenvectors $\mathbf u_{L+1},\cdots,\mathbf u_{N}$ are mutually orthogonal, and that $|\alpha| \approx 1$ we can write \eqref{eq:karmar_music} as
	
	\begin{align}
	\frac{1}{N}\mathbf d^H(z) \mathbf Q_n \mathbf Q_n^H \mathbf d(z)
	\end{align} 
	which completes the proof.
\end{proof}

The formation of the root-MUSIC polynomial is therefore optimal with respect to the polynomial coefficients encoded in the noise eigenvectors. What then remains is the root-selection algorithm, in other words, finding the minimum of the function $\boldsymbol{\epsilon}_{\rm min}$ with respect to $\alpha$. This corresponds to finding the values of $\alpha$ which correspond to the location of the sources. Writing $\alpha = u + jv, u,v \in \mathbb{R}$, $\boldsymbol{\epsilon}_{\rm min}$ becomes a real polynomial in two variables, which attains its minima at its stationary points
\begin{align}
\frac{\partial \boldsymbol{\epsilon}_{\rm min}}{\partial u} = 0,\    & \frac{\partial \boldsymbol{\epsilon}_{\rm min}}{\partial v} = 0. \label{eq:opt_cond}.
\end{align}
Fig. \ref{fi:AlgebraicGeometryWorks} depicts the intersection of the two algebraic plane curves described by equations \eqref{eq:opt_cond} along with the roots that correspond to the actual target locations. The solution of the above stationary points is investigated in Section \ref{subsec:karmarkar}. A certificate of the satisfaction of a necessary condition for optimality is produced.

\begin{figure}[t]
	\begin{center}
		\includegraphics[width=\linewidth,trim=0 250 0 250]{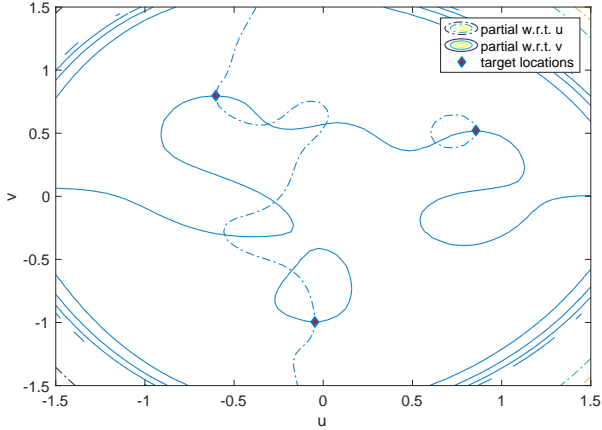}
	\end{center}
	\vspace{2ex}
	\caption{The 0 level cut of equations \eqref{eq:opt_cond} corresponding to the last two eigenvectors of the SCM corresponding to a system with 10 antennas, with 3 targets impinging on the array. The real intersection points clearly correspond to target locations.}
	\label{fi:AlgebraicGeometryWorks}
\end{figure}

\section{Algorithms}
\label{sec:Algorithms}
Several algorithms have been developed in the computer algebra literature to compute the approximate GCD of polynomials whose coefficients are either only imprecisely known, are perturbed by noise, or are in some other sense ``close'' to having a GCD, but for which Euclid's algorithm will fail to find an appropriate GCD \cite{KarmarLakshman}. In \cite{Pan}, a ``root-clustering'' algorithm was proposed whereby the roots of two polynomials $f(x)$ and $g(x)$ were compared pairwise, and deemed to have a common root if a pair of roots were found to exist within a radius $\delta$ of each other. Once a root pair $\{ x_{q,1}$, $x_{q,2} \}$ is identified, the corresponding root of the approximate GCD, $y_q$, is computed as the average of the two points. Finally, the approximate GCD, $h(x)$, is calculated as
\begin{align}
h(x) = \prod_{q = 1}^{Q} (x-y_q)
\end{align}
where $Q$ is in the number of discovered paired roots.

We extend this approach here to the array processing scenario considered, where additionally we will consider multiple polynomials instead of just two. Moreover, when calculating the ``true'' root we make use of the prior information that the roots corresponding to the sources must lie on the unit circle. The number of roots discovered in a cluster will be shown to be useful for determining the number of sources as well as for separating closely located sources, thus, addressing the open important problems of joint estimation of the number od sources and source parameters and non-idealities in the date model. We detail this method in Subsection~\ref{subsec:cluster}.

Existing approximate GCD algorithms, including the one highlighted above, provide a good estimate of the ``true'' GCD in most circumstances. However, they are based on heuristics and as such do not provide guarantees for how far the estimated GCD will be from the ``nearest'' GCD, where nearness is defined in the following way. Given a system of polynomials that does not have a non-trivial GCD, and therefore has no solution, the coefficients of the polynomials in the system can be perturbed such that the perturbed system does have a non-trivial GCD. The GCD corresponding to the system with minimum total perturbation from the original system is defined then to be the nearest GCD. As was shown in the previous section, this corresponds to the formation of the root-MUSIC polynomial, and solution of \eqref{eq:rMUSIC}. In this section we derive an algorithm which solves the nearest GCD problem with respect to the polynomial coefficients. As a result, we provide a certificate of a necessary condition for optimality.

The problem of finding the ``nearest'' GCD in the aforementioned sense is solved in 
\cite{KarmarLakshman} and \cite{Zeng},\cite{zeng1}. The authors there minimize the perturbation of a system of polynomials using the method of Lagrange multipliers subject to the perturbed system having a non-trivial GCD. It is shown that the solution to this problem is equivalent to finding the real intersection points of two algebraic plane curves. The author of \cite{Zeng},\cite{zeng1} solves this problem using only matrix-vector operations which are backward stable.\footnote{Thus, resolving the issues with Wilkinson's type polynomials.}

\subsection{Root Clustering Method}
\label{subsec:cluster}
The root clustering approach of \cite{Pan} proceeds in two steps. In the first step, roots that are ``close'' are paired together. Then in the second step, a ``consensus'' root is calculated based on each of these roots. 

For our problem here, since we consider more than two polynomials, we adopt a different clustering approach. There is a wide variety of clustering algorithms within the literature of machine learning \cite{Tibshirani}. Because we are looking for tight clusters, we adopt the hierarchical approach of agglomerative clustering. Agglomerative clustering starts by considering each data point as its own cluster. At each step, it calculates for each cluster the closest neighboring data point. The cluster which has the closest neighbor is then merged with this neighbor to form a new cluster, while the distance between the original cluster and the new member is stored. The distance by which the roots are compared is called a ``dissimilarity'', and it can, in general, be any pseudonorm. This process continues until all the data points are grouped in one large cluster, with all clusters, $c_i$, are ranked by their dissimilarity in a tree-structure known as a ``Dendrogram''.
\begin{figure}[t]
	\begin{center}
		\includegraphics[width=\linewidth,trim=0 250 0 215]{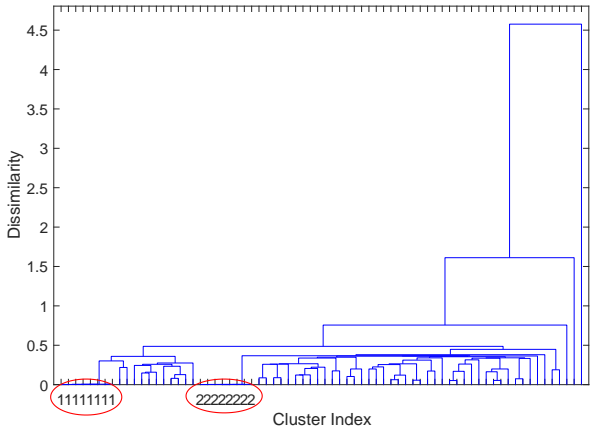}
	\end{center}
	\vspace{2ex}
	\caption{Root dendrogram corresponding to ULA with $N = 10$ elements, $L = 2$ targets at $30^o$ and $-40^o$. Target clusters are highlighted with red ellipses.}
	\label{fig:dendrogram}
\end{figure}

{\it Illustrative Example 2:} Fig.~\ref{fig:dendrogram} shows a dendrogram for a system consisting of a ULA with $N = 10$ elements and $L = 2$ sources impinging from $30^o$ and $-40^o$. A total of $T = 100$ snapshots were collected at an SNR of $10$dB. Two clusters of $8$ roots each with a dissimilarity close to $0$ can be easily observed in the figure and are marked by ellipses.

\begin{figure}[t]
	\begin{center}
		\includegraphics[width=\linewidth,trim=0 250 0 215]{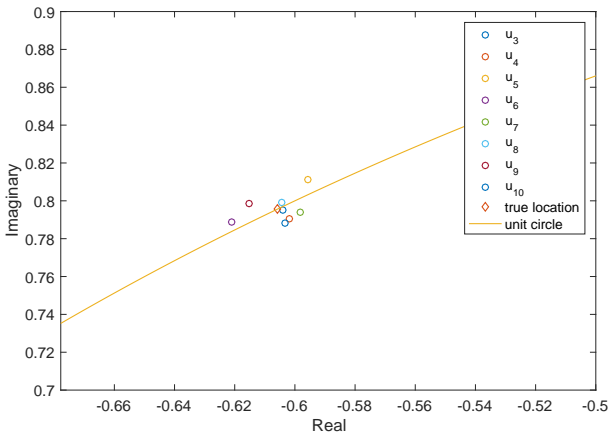}
	\end{center}
	\vspace{2ex}
	\caption{Root scatter corresponding to ULA with $N = 10$ elements, $L = 2$ targets at $31^o$ and $45^o$. Note that there are $N - L = 8$ roots in the cluster.}
	\label{fig:rootscatter}
\end{figure}

The number of ``tight'' clusters found in this process gives us the estimate of the number of sources $\hat{L}$. We note that the roots being clustered can be taken from either all eigenvectors or from a noise-subspace estimate, such as the one produced from a eigenvalue based detection method. Once the root clusters $c_i = \{|z_{i,1},\cdots,z_{i,J_i}\}$ are obtained, the clusters are treated as samples of a distribution parametrized by the ``true'' root. The remaining challenge is then to estimate the true root from the sample. In \cite{Pan}, this is just taken to be the centroid of the root cluster. However, in the context of the considered array processing scenario , the prior information that the root estimate must lie on the unit circle can be exploited. 

{\it Illustrative Example 3:} In Fig. \ref{fig:rootscatter}, the roots corresponding to the last $8$ eigenvectors of the sample covariance matrix $\hat{\mathbf R}_{xx}$ obtained for the same scenario considered in Illistrative Example.$~2$ are plotted, along with the unit circle and the target location. Each mark on the x-axis of Fig.~ \ref{fig:dendrogram} is a root of an eigenvector of $\mathbf R_{xx}$. If there is a bracket between two roots, then these roots are in a cluster. If there is a bracket between clusters, these two clusters are part of a larger cluster. The height of the bracket indicates the maximum dissimilarity between elements in that cluster. From Fig. $3$ it is clear that large, tight clusters are rare events which only exist as a result of the algebraic structure of the noise subspace. Thus this property is useful from a detection and estimation standpoint.

Since the estimate of the source location is based on the phase argument of the root in question, only distortions tangential to the unit circle contribute to estimation error \cite{BhaskarRao}. Radial distortions do not cause any error, and thus, the roots of the noise eigenvectors can be projected onto the unit circle by simply normalizing them by their corresponding absolute values. The estimate of the source location $\alpha_l$ is then given by the Fr\'echet mean of the roots in each cluster, as the unit circle is a Riemannian manifold. To calculate a Fr\'echet mean, an iterative procedure is required whereby the points in question are mapped onto an Euclidian tangent space (exponentiation) in which the mean of these points is calculated. This mean estimate is then mapped back onto the manifold (logarithm), and serves as the tangent point for the next iteration. For a detailed explanation on the theory and computation of Fr\'echet means we refer the reader to \cite{ArnauBarbarYang}. However, instead of calculating a Fr\'echet mean for estimating the source location, a simplear procedure can be used. In fact, it is apparent from Fig.~\ref{fig:rootscatter} that, since the roots are closely located, the ``small angle'' approximation can be used instead while incurring very little error. Thus, a ``good consensus'' root estimate is just $e^{-j\phi_{avg}}$ where $\phi_{avg}$ is the mean phase argument of the roots. The estimate of the DOA is then given by $\rm{sin} \left( \frac{1}{\pi} \phi_{avg,l} \right)$, where $\phi_{avg,l}$ is the average phase argument of the $l$-th cluster of roots, for $l \in 1,\cdots,\hat{L}$. These steps of estimating $\hat L$ and then the roots are summarized in Algorithm~$1$, where the input argument $\delta$ is 

\begin{algorithm}
	\label{alg:rootcluster}
	\begin{algorithmic}[1]
		\Procedure{Algorithm 1}{$\mathbf Q$, $\delta$}
		\State Compute roots of columns of $\mathbf Q \to \mathbf r$
		\State Agglomerative clustering on $\mathbf r$ into clusters $d(c_i) < \delta,\ c_i = \{z_{i,1},\cdots,z_{i,J_i}\}$
		\State $\hat{L} = |\{c_i,\forall i \  | \ J_i > 2 \}|$
		\For{$i \leq \hat{L}$}
		\State $\phi_{\mathrm{avg},i} = \frac{1}{|c_i|} \sum_{j=1}^{J_i} \angle z_{i,j}$
		\State $\theta_i = \rm{sin}^{-1} \left( \frac{1}{\pi} \phi_{\mathrm{avg},\mathit{i}} \right)$
		\EndFor
		\EndProcedure
	\end{algorithmic}
\end{algorithm}

Assuming that there are $L$ distinct targets, there must be $N-L$ eigenvectors of the form \eqref{eq:epsilon}. Thus, the roots corresponding to the source locations are estimated form $L$ tight clusters of $N-L$ roots each. Thus, not only do the roots contain information about the source location, but the number of roots in each cluster contains information about the number of sources. It is this information which can be then used to separate closely located sources. Specifically, if two sources are closely located, they will be grouped into the same cluster. According to \eqref{eq:epsilon}, a valid cluster cannot have more than $N-1$ roots, since the signal must occupy a subspace with at least dimension $1$ to which the resulting $N-1$ dimensional subspace must be orthogonal. Thus, tight clusters with more than $N-1$ roots are deemed to correspond to multiple closely located targets. Accounting for this possibility of resolving closely located sources, Algorithm 1 can be extended to exploit the above explained criterion to jointly estimate the number of sources and their locations. This results in Algorithm 2.

\begin{algorithm}
	\begin{algorithmic}[2]
		\Procedure{Root Clustering Algorithm}{$c_i$}
		\While{$|c_i| > N-1$}
		\State $\phi_{avg} = \frac{1}{J_i} \sum_{j=1}^{J_i} \angle z_{i,j}$
		\For{$j \leq J_i$}
		\If{$\angle z_{i,j} \leq \phi_{avg}$}
		$\angle z_{i,j} \to c_{i,1},\  c_{i,1} = c_{i,1}\cup \{\angle z_{i,j}\}$
		\Else {$\angle z_{i,j} \to c_{i,2},$} $\ c_{i,2} = c_{i,2} \cup \{z_{i,j}\}$
		\EndIf
		\State $\phi_{\mathrm{avg},i,k} = \frac{1}{|c_{i,k}|} \sum_{k=1}^{|c_{i,k}|} c_{i,k}$
		\State $\theta_{c_i} = \rm{sin}^{-1} \left( \frac{1}{\pi} \phi_{\mathrm{avg},i} \right)$
		\EndFor
		\EndWhile
		\EndProcedure
	\end{algorithmic}
\end{algorithm}

Algorithm 2 accepts as an input a cluster of roots. First, the algorithm checks whether the number of roots in the detected cluster is consistent with \eqref{eq:epsilon}. If the number of roots in the cluster is larger than $N-1$ the cluster is ``split.'' That is, provided the average phase argument of the roots in the cluster, the roots are grouped into those whose phase argument is larger, or less than the average phase argument. The algorithm then returns two DOAs estimated from the average phase argument of the two groups of roots.

The dominant task in both Algorithm $1$ and Algorithm $2$ in terms of computational complexity is the root-clustering algorithm. Given $N$ nodes, the $N^2$ pair-wise distances between nodes are computed, and stored in a sorted array of ``next-best'' distances. When two nodes are merged into a single cluster, the ``next-best'' distance array is updated, and the clustering algorithm continues on the basis of this new updated array. Each update can be done with complexity $\mathcal{O}(N)$, while the computation of the initial array requires complexity $\mathcal{O}(N^2)$. Thus, the overall complexity of the algorithm is $\mathcal{O}(N^2)$. If there are $K$ sources, $N-K$ noise eigenvectors are selected, each having $N-1$ roots. Thus the overall complexity of the root clustering algorithm, in terms of $N$ and $K$ is $\mathcal{O}((N^2 - N(K+1) + K)^2)$ in which $N^4$ dominates, resulting in a complexity bound $\mathcal{O}(N^4)$.

The high performance of these algorithms will be demonstrated in Section~\ref{Simul}.

\subsection{Nearest GCD}
\label{subsec:karmarkar}

In \cite{KarmarLakshman}, the nearest GCD algorithm proceeds in two steps. First, the points $\alpha$ are found, then $\lambda_i$ and $\mu_i$ are calculated from \eqref{eq:lagrange1} and \eqref{eq:lagrange2}, and added to the coefficients of $f(z)$ and $g(z)$ in order to recover the system of polynomials which has a non-trivial GCD. The problem of root-selection in root-MUSIC can thus be expressed as finding the real intersection points of two algebraic plane curves. 

There are algorithms to solve this problem in $N^3$ operations, however they all require symbolic computations. We instead adopt the approach of \cite{Zeng} to the problem posed in \cite{KarmarLakshman} which requires only numeric matrix-vector calculations and Gauss-Newton iterations. The derivation that follows is an exposition of the proofs in \cite{Zeng}. For details and complete proofs, we direct the reader to \cite{Zeng}.

The derivation of the algorithm in \cite{Zeng} follows naturally from the definition of the Sylvester matrix. Consider two polynomials $f(x)$ and $g(x)$ of degrees $m$ and $n$ respectively with a GCD $u(x)$ of degree $k < \ m,n$. Let the cofactors of $f(x)$ and $g(x)$ be denoted as $v(x)$ and $w(x)$ respectively. Recalling the definition of convolution matrices from Lemma $1$ we can write the identity $f(x)w(x) - g(x)u(x) = 0$ as

\begin{align}
\mathbf B_{n-k}(f) \mathbf w - \mathbf B_{m-k}(g) \mathbf v = \boldsymbol{0}
\end{align} 
where $m-k$ and $n-k$ are the degrees of the cofactors $v(x)$ and $w(x)$ respectively, and $\mathbf B_{n-k}(f)$ is the convolution matrix defined by the coefficients of $f$ with $n-k+1$ columns. The identity can then be rewritten as
\begin{align}
[\mathbf B_{n-k}(f) | \mathbf B_{m-k}(g)] \begin{bmatrix} \mathbf w\\
-\mathbf v
\end{bmatrix} = \mathbf S_k(f,g) \begin{bmatrix} \mathbf w \\
-\mathbf v
\end{bmatrix} = \boldsymbol{0}. \label{eq:sylv}
\end{align}
The matrix $\mathbf S$ is the known as the Sylvester matrix of polynomials $f$ and $g$ and is singular if and only if the two polynomials have a non-trivial GCD. This can be straight forwardly shown from the original identity via Bezout's theorem. The degree $k$ in the above expression can vary from $1$ to $\rm{deg}(u(x))$ resulting in different Sylvester matrices. The different Sylvester matrices have different nullities. Specifically $\mathbf S_{1}(f,g)$ has nullity equal to $k$, and $\mathbf S_{k}(f,g)$ has nullity equal to one. Thus, the cofactors $w(x)$ and $v(x)$ can be solved for by identifying the kernel of $\mathbf S_{k}(f,g)$ \cite{Zeng}. Once one has the cofactors $\mathbf w$ and $\mathbf v$, the GCD $\mathbf u$ can be found as the solution to the simultaneous linear system

\begin{align}
\mathbf B_k(v) \mathbf u = \mathbf f,\   \mathbf B_k(w) \mathbf u = \mathbf g. 
\end{align}
The estimate of the approximate GCD $\mathbf u$ in \cite{Zeng} is thus given as the solution to the following linear system 
\begin{align}
\begin{bmatrix}
\mathbf r^H \mathbf u - 1 \\
\mathbf B_k(v) \mathbf u \\
\mathbf B_k(w) \mathbf u
\end{bmatrix} = \begin{bmatrix}
0 \\
\mathbf f \\
\mathbf g
\end{bmatrix} \label{eq:zeng}
\end{align}
for $\mathbf u$. In the above system $\mathbf r$ is pre-defined a scaling vector. For example, if we wish $\mathbf u$ to be monic, $\mathbf r$ would be the vector $[1,0,\cdots,0]^T$. In the simulations, we set $\mathbf r$ to be equal to the initial estimate of the GCD, $\mathbf u_0$. Denoting the left hand side of \eqref{eq:zeng} as $\mathbf f(\mathbf z)$ where $\mathbf z = [\mathbf u, \mathbf v, \mathbf w]^T$ and the right side as $\mathbf b$ the solution sought is the least squares solution which minimizes the distance $\| \mathbf F - \mathbf b \|_2$. This is done via a Gauss-Newton iteration. The Jacobian of $\mathbf f(z)$ has a closed form expression in terms of the scaling vector $\mathbf r$ and convolution matrices $\mathbf B_k$ given as
\begin{align}
\mathbf J(\mathbf z) &= \begin{bmatrix}
\mathbf r^H & & \\
\mathbf B_k(v) & \mathbf B_{n-k}(u) & \\
\mathbf B_k(w) & & \mathbf B_{n-k}(u)
\end{bmatrix}
\end{align}
which leads to the following Gauss-Newton iteration
\begin{align}
\mathbf z_j = \mathbf z_{j-1} + \mathbf J(\mathbf z_{j-1})^{\dagger} [\mathbf f(\mathbf z_{j-1}) - \mathbf b ] \label{eq:gauss_newton}
\end{align}
where $(\cdot)^{\dagger}$ denotes the Moore-Penrose pseudo-inverse. With these defined the approximate GCD algorithm of \cite{Zeng} proceeds in two steps. In the first step, the degree of the approximate GCD,$\ k$, is estimated via the analysis of the nullity of a sequence of Sylvester matrices of varying $k$. Then, with the estimate $\hat{k}$, the Gauss-Newton iteration \eqref{eq:gauss_newton} is run until convergence in order to minimize the distance between the estimated polynomials $\hat{\mathbf f},\hat{\mathbf g}$ which have a non-trivial GCD by definition, and the observed polynomials $\mathbf f,\mathbf g$. The initial Sylvester matrix $\mathbf S_0(f,g)$ is defined as $[\mathbf f | \mathbf g]$ assuming, as in our case, that $f$ and $g$ have the same degree. In the exact GCD algorithm $S_k(f,g)$ is full column rank for all $k < n - \mathrm{deg}(u(x))$. Thus, one constructs a sequence of Sylvester matrices $S_k(f,g)$ until a singular $S_k(f,g)$ is found, after which the cofactors and $\mathbf u$ can be found using \eqref{eq:zeng}. In the approximate GCD algorithm, $S_k(f,g)$ will only be approximately singular. Thus a threshold must be set to decide whether a singular value of $\mathbf S_k(f,g)$ is close enough to $0$. This is defined in terms of an $\epsilon$ such that

\begin{align}
\boldsymbol{\epsilon} = \left \| \begin{bmatrix}\hat{\mathbf f} \\ \hat{\mathbf g}\end{bmatrix} - \begin{bmatrix}\mathbf f \\ \mathbf g\end{bmatrix}\right \|_2 < \zeta. \label{eq:eps}
\end{align}
Given such a $\zeta$, $\mathbf S_k (f,g)$ is said to be singular if it has a singular value $\sigma_{k,\mathrm{min}} < \epsilon \sqrt{2k - 2}$. If $\mathbf S_k(f,g)$ is approximately singular with respect to $\zeta$, then it is possible that $f$ and $g$ have an approximate GCD which satisfies \eqref{eq:eps}. Once a matrix $\mathbf S_k(f,g)$ is found with $\sigma_{k,\mathrm{min}} < \zeta$, the linear system and Gauss-Newton iterations \eqref{eq:zeng}, \eqref{eq:gauss_newton} are run to find an polynomial pair $\hat{\mathbf f}$ and $\hat{\mathbf g}$ with a non-trivial GCD of degree $n - k$. If \eqref{eq:eps} is satisfied, the solution to \eqref{eq:gauss_newton} $\mathbf u$ is returned as the approximate GCD. If not, then $k$ is increased and an approximate GCD is sought of degree $n-k-1$. This top down approach guarantees that the approximate GCD of maximal degree with respect to epsilon is found. At convergence of \eqref{eq:gauss_newton} we have the following condition.

\begin{align}
\mathbf J(\mathbf z)^{\dagger}(\mathbf f(\mathbf z) - \mathbf b) = 0 \label{eq:stationary}
\end{align}

This does not guarantee that the polynomial pair is of minimal distance from the manifold of polynomial pairs with GCD degree equal to $n-k$ to the observed polynomial pair $f,g$. However, it is a necessary condition for the nearest polynomial pair on this manifold. In cases of small $\epsilon$ - which in our context corresponds to high SNRs - it becomes more and more likely that satisfaction of \eqref{eq:stationary} ensures optimality. The satisfaction of \eqref{eq:stationary} is thus a certificate of a necessary condition of optimality.

The first step of this algorithm is merely a method to find an initial GCD estimate $\mathbf u$ and set of cofactors $\mathbf v$. Of course, given any initial estimate $\mathbf u$, it's optimality (or sub-optimality) can be verified and refined using the iteration \eqref{eq:gauss_newton}. Given initial estimates $\widehat{\alpha_l}$ of the source locations, one can produce the initial estimate $\mathbf u_0$ using the formula \eqref{eq:generator}. The cofactors corresponding to each eigenvector can then be found by solving the overdetermined system
\begin{align}
\mathbf B_{n-k}(u) \mathbf v_i &= \mathbf q_i
\end{align}
for $\mathbf v_i$. The matrices $\mathbf B_k (v_i)$ can then be formed from $\mathbf v_i$. The performance of this approach is verified in Sec.~\ref{Simul}, where the initial estimates are provided by the root-clustering algorithm.   

A detailed analysis of the convergence of \eqref{eq:gauss_newton} is given in \cite{Zeng},\cite{zeng1}. We only repeat the result that if the polynomial pair has an approximate GCD which satisfies \eqref{eq:eps}, then there exists some small positive real number $\mu$ such that \eqref{eq:gauss_newton} converges to the approximate GCD from any initial state $\mathbf z_0$ within distance $\mu$ of the observed polynomial pair. The total complexity of the root-certificate algorithm is $k \cdot \mathcal{O}((N(N-K)+1)^3)$, where $k$ is the total number of iterations required for convergence, $N$ is the number of antennas, and $K$ is the number of sources, resulting in a complexity bound of $\mathcal{O}(N^6)$. The total number of iterations is not known a-priori, and is thus investigated in \ref{Simul}.

From the definitions in \ref{sec:SelecCriterion} the estimate of the number of targets is then the degree of $u(x)$ and the estimates of the source locations are its roots. The aforementioned steps are summarized in Algorithm 3 (root-certificate algorithm).

\begin{algorithm}
	\label{alg:lagrange}
	\begin{algorithmic}[1]
		\Procedure{Root Certificate (uvGCD \cite{Zeng})}{$\mathbf q_i, \mathbf q_j$, $\epsilon$}
		\State $\mathbf S_0 (q_i,q_j) = \mathbf Q_0 \mathbf R_0$
		\For{$k = 0,...,N-2$}
		\State Obtain $\sigma_k$ from $\mathbf R_k$
		\If{$\sigma_k \leq \epsilon \sqrt{2k - 2}$}
		\State Solve $\mathbf S_k(f,g) \begin{bmatrix} \mathbf w \\
		-\mathbf v
		\end{bmatrix} = \boldsymbol{0}$ for $\mathbf v$ and $\mathbf w$
		\State Solve $\mathbf F(\mathbf z) - \mathbf b = 0$ for $\mathbf u$
		\State Iterate \eqref{eq:gauss_newton} until $\mathbf J(\mathbf z)^{\dagger}(\mathbf F(\mathbf z) - \mathbf b) = 0$
		\If{$\|[(\hat{\mathbf q_i}, \hat{\mathbf q_j}) - (\mathbf q_i,  \mathbf q_j)\|_2 \leq \epsilon$}
		\State Return $\mathbf u$ terminate Algorithm $3$
		\EndIf
		\EndIf
		\EndFor
		\State Return $\mathbf u = [1]$
		\EndProcedure
	\end{algorithmic}
\end{algorithm}

\subsection{Parameter Selection}

There is a physical sense in which to select the relevant parameters in Algorithms $1$ and $3$. The dissimilarity in Algorithm $1$ and the ball of radius $\zeta$ in Algorithm $3$ represent the perturbation of the roots of the eigenvectors $\mathbf u_k$. As these roots are each injectively related to the parameter being estimated, $\theta$, the perturbations themselves must be on the order of the estimation error given the SNR conditions. In order to address the selection of parameters for the root certification algorithm, we must relate uncertainty of the target locations to the perturbation of the coefficients of the noise eigenvectors. In essence, the Gauss-Newton iteration searches in a ball of radius $\zeta$ on the GCD manifold of dimension $k$, whereas the estimation error of the algorithm is quantified in terms of the root estimates themselves. To address this, we use a result proved in reference $[20]$. Let $u(x)$ and $\tilde{u}(x)$ be two polynomials of the same degree $n$, whose roots $\alpha_i$ and $\tilde{\alpha}_i$ are such that
\begin{align}
|\alpha_i - \tilde{\alpha}_i| \leq \delta
\end{align}
then $\|u(x) - \tilde{u}(x)\| \leq \|u(x)\| \cdot ((1+\delta)^n - 1)$ where $\| \cdot \|$ indicates, in this case, the $2$-norm of the polynomial coefficient vector. The resulting relation of the maximum pairwise root distances and the perturbation of coefficients allows us to relate target location uncertainty with the radius of search on the GCD manifold. Let $\triangle \theta$ be the maximum uncertainty of the target locations. Then, from the definition of our polynomial roots, we express the difference as
\begin{align}
\delta =&|e^{j \pi \rm{sin}(\theta_1)} - e^{j \pi \rm{sin}(\theta_1 - \triangle \theta)}| \nonumber \\
= &|e^{j \pi \rm{sin}(\theta_1)} - e^{j \pi \rm{sin}(\theta_1)\rm{cos}(\triangle \theta) - \rm{cos}(\theta_1) \rm{sin}(\triangle \theta)}| . \label{eq:approx2}
\end{align}

Using the small angle approximation we, rewrite \eqref{eq:approx2} as
\begin{align}
\delta \approx&|e^{j \pi \mathrm{sin}(\theta_1)} \cdot (1 - e^{- j \pi \mathrm{cos}(\theta_1) \mathrm{sin}(\triangle \theta)})| \nonumber \\
\leq &|1 - e^{-j \pi sin{\triangle \theta}}| \approx |1 - e^{-j \pi \triangle \theta}|
\end{align} 

{\it Illustrative Example 4:} Say we know the SNR conditions under which we are working, and as a result, we know that the uncertainty in our target estimates is $0.5^o$. Using the above relation, we set $\delta = 0.0274$. Assuming $K = 3$ targets, and $\|u(x)\| = 1$ we have a maximum perturbation of $\zeta = (1 + \delta)^3 - 1 = 0.0845$. We stress that this is a guideline for parameter selection, but on the basis of this analysis the suggested baseline is $\zeta = \mathcal{O}(10^{-2})$.

\section{Simulation Results}
\label{Simul}
To examine the performance of the proposed algorithms we consider two scenarios. The first, and most challenging, scenario is the case where we have multiple sources that are closely located. The second scenario is when we have multiple widely spaced sources. We consider two performance metrics. The first is the probability of detecting the correct number of sources. The second is the root mean squared error (RMSE) performance of the DOA estimates. Performance comparisons for the algorithms described in this paper are given with respect to AIC, MDL, and the Sequentially Rejective Bonferroni Procedure (SRBP) described in \cite{Zoubir}. To implement the SRBP, jack-knife bias correction is applied to both the eigenvalues of the SCM corresponding to the all $T$ observations, and the sets of bootstrapped eigenvalues. A total of $B = 200$ booststraps are collected of size $M = 30$ each. A global significance level of $\beta = 0.03$ is maintained for both scenarios. In addition to these two metrics, we provide histograms of the source number estimates. The root-certificate algorithm is implemented using the software uvGCD which is made freely available by the author of \cite{Zeng}. The Gauss-Newton iteration \eqref{eq:gauss_newton} was implemented by the authors. In all simulations we assume that the sources are impinging on a ULA of $N = 10$ elements, with element spacing $\lambda_c/2$, and $T = 100$ snapshots are collected. The sources are assumed to be equal power i.i.d. Gaussian zero mean sources with covariance $\sigma_s^2 \mathbf I$, and the noise assumed to be i.i.d. Gaussian with zero mean and covariance $\sigma_n^2 \mathbf I$. Throughout all simulations, the signal to noise ratio is varied from $-15$dB to $30$dB in $3$dB increments.

\subsection{Widely Separated Sources}

In terms of estimation and detection performance, our first point of inquiry is whether the proposed algorithms perform well in normal situations where subspace based DOA estimation algorithms are known to provide optimal or near-optimal performance. To perform this comparison, we consider the situation when two sources impinge on the array from directions $[-10^o\ 20^o]$ respectively. In order to test the RMSE performance of the root-clustering (RClA) and root-certificate algorithms (RCA) we compare them to Root-MUSIC. The RClA adaptively learns the number of sources from the root structure of the noise eigenvectors, root-MUSIC learns the number of sources via $MDL$ and the RCA estimates the number of sources from the GCD manifold of maximum co-dimension corresponding to the last two eigenvectors of $\mathbf R_{xx}$.

\begin{figure}[t]
	\vspace{4ex}
	\begin{center}
		\includegraphics[width=\linewidth,trim=0 250 0 250]{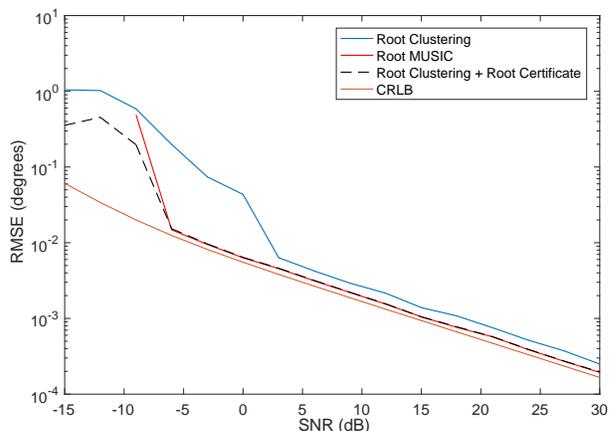}
	\end{center}
	\vspace{2ex}
	\caption{RMSE comparison of RClA and RCA with root-MUSIC in widely separated sources scenario. Two sources impinge on the array from directions $\Theta = [-10^o \ 20^o]$ respectively.}
	\label{fi:wide_rmse}
\end{figure}

Fig. \ref{fi:wide_rmse} depicts the performance comparison between RClA and RCA and root-MUSIC with respect to RMSE. 
As one expects, root-MUSIC performs near-optimally from $-3$dB SNR onwards. Algorithm $2$ converges later but still performs well in a usable and practical SNR range from $0$dB to $10$dB. It's worth noting that the gap between RClA and root-MUSIC never closes at high SNRs. This is likely a consequence of Theorem $1$. However, once RCA is applied to the initial estimate provided by RClA the gap closes between root-MUSIC and root-clustering. In sum, the proposed methods sacrifice no estimation performance over the whole usable SNR range of subspace based methods. The only price is thus increased computational complexity. 

Fig. \ref{fi:wide_detect} depicts the detection performance of RClA and RCA in comparison with MDL and AIC. AIC detects the correct number of sources before any other method, but this is due mostly to the tendency of AIC to overestimate the number of sources. MDL correctly estimates the number of sources with $100\%$ accuracy starting from $-6$dB followed by RClA which estimates the correct number of sources with $97\%$ starting from $-3$dB, after which perfect detection is observed. Neither RCA nor AIC ever achieve perfect detection, however, RCA considerably outperforms AIC in high SNR conditions. The performance of the SRBP tracks very closely to that of MDL (as is predicted by theory in large sample sizes), and that of the root-clustering method in the low SNR regime. However, at high SNRs, the SRBP never achieves $100 \%$ at this significance level. This is due to over-estimation. One could restrict the significance level to achieve $100 \%$ correct detection at high SNRs, but this necessarily will induce under-estimation in the low-SNR regime, in which it is likely that fewer null-hypotheses are able to be rejected at any significance level. As Fig. \ref{fi:hist_wide} demonstrates, underestimation at the significance level $\beta = 0.03$ is already a problem for the SRBP at $-6$dB SNR.

\begin{figure}[t]
	\vspace{4ex}
	\begin{center}
		\includegraphics[width=\linewidth,trim=0 250 0 250]{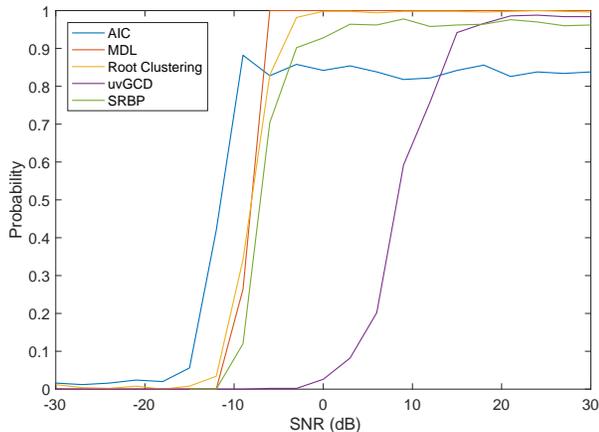}
	\end{center}
	\vspace{2ex}
	\caption{Probability of detection of correct number of sources comparison for AIC, MDL, and Algorithms $2$ and $3$.}
	\label{fi:wide_detect}
\end{figure}

\begin{figure}[t]
	\vspace{4ex}
	\begin{center}
		\includegraphics[width=\linewidth,trim=0 250 0 250]{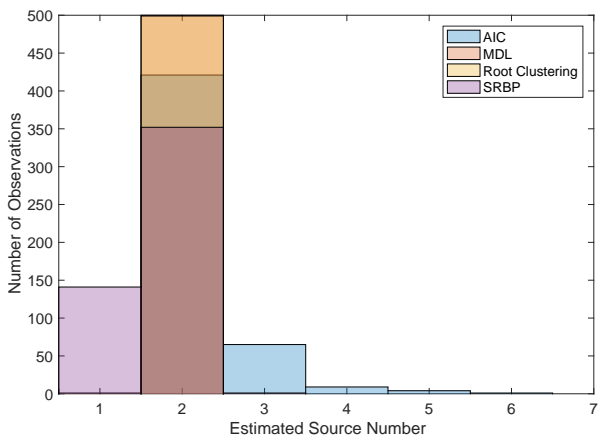}
	\end{center}
	\vspace{2ex}
	\caption{Histogram comparison of source number estimates of AIC, MDL, SRBP, and RClA for SNR $= 0$dB and two sources impinging from $\Theta = [-10^o\ 20^o]$ respectively.} 
	\label{fi:hist_wide}
\end{figure}

Fig. \ref{fi:hist_wide} depicts the histogram of source estimates for RClA in comparison with MDL and AIC at $0$dB SNR. RCA is excluded from the figure for the reason that it never corrects the correct number of sources in this scenario. RClA correctly identified the number of sources in $499$ out of $500$ Monte Carlo trials. More importantly, RClA never overestimates the number of sources, while AIC sometimes wildly overestimates the number of sources. 

\subsection{Closely Located Sources}

We consider the case illustrated in Example $1$ earlier. Two equal power sources impinge on the array from directions $[31^o\ 32^o]$. RClA is employed to first detect the number of sources and then estimate their DOAs. In order to test the detection performance, we compare RClA to both AIC and MDL. In order to test the RMSE performance, we compare RClA to root-MUSIC which has been provided with the correct number of targets for all SNRs. RClA, by comparison, learns the number of targets before performing the DOA estimation.  

\begin{figure}[t]
	\vspace{4ex}
	\begin{center}
		\includegraphics[width=\linewidth,trim=0 250 0 250]{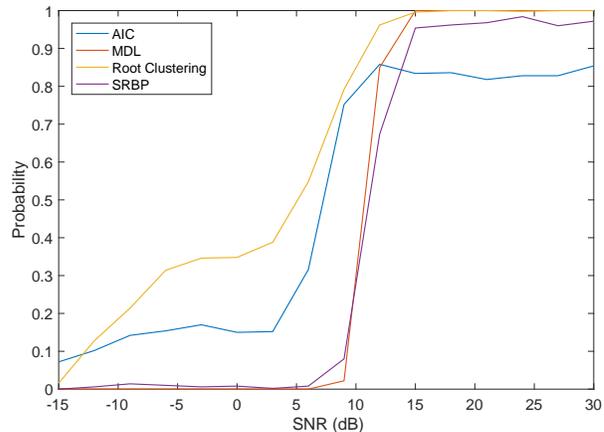}
	\end{center}
	\vspace{2ex}
	\caption{Performance comparison between AIC, MDL, and RClA in terms of probability of correctly learning the number of sources in the case of two closely located sources at $\Theta = [31^o\ 32^o]$.}
	\label{fi:prob_aicmdl}
\end{figure}

Fig. \ref{fi:prob_aicmdl} demonstrates the detection performance of RClA while SNR is varied. We notice two separate performance gaps between the proposed method and AIC and MDL. As previously noted, AIC has a tendency to overestimate the number of sources. This explains the fact that AIC never detects the correct number of targets 100$\%$ of the time. As such, our algorithm enjoys a permanent performance advantage over AIC in high SNR scenarios. By contrast, MDL tends to underestimate sources. In Fig. \ref{fi:prob_aicmdl} we observe a rapid transition in detection probability for MDL. Before approximately $6$dB SNR, MDL never correctly estimates the number of sources. Thus, Algorithm $2$ enjoys a significant performance advantage in the low SNR region. By $12$dB, however, MDL correctly estimates the number of sources $100 \%$ of the time. The performance of the SRBP tracks very closely the performance of MDL. However, the root clustering algorithm is able to resolve the targets at a much lower SNR, and perfectly at high SNRs.

\begin{figure}[t]
	\vspace{4ex}
	\begin{center}
		\includegraphics[width=\linewidth,trim=0 250 0 250]{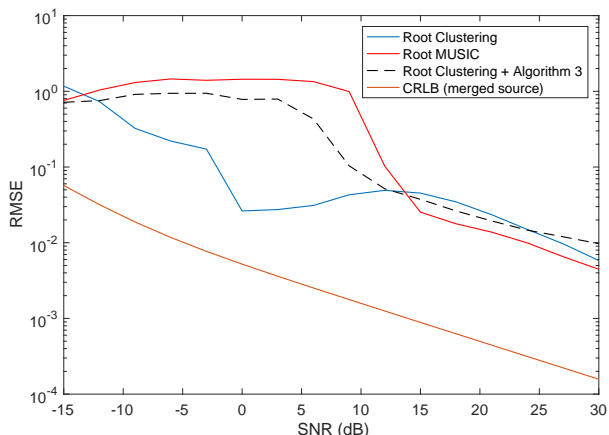}
	\end{center}
	\vspace{2ex}
	\caption{Comparison between RClA, and RCA and root-MUSIC for closely targets located at $\Theta = [31^o\ 32^o]$.}
	\label{fi:Multisource}
\end{figure}

If it is possible to estimate the SNR conditions in which the receiver is operating, the idea of using AIC in low SNR conditions and MDL in high SNR conditions may seem natural. However, as we see in Figs. \ref{fi:hist_wide} and \ref{fi:hist} AIC significantly overestimates the number of sources. Morever, as we see in Fig. \ref{fi:Multisource}, even when provided with the correct number of sources, in low to moderate SNR regions, root-MUSIC is not able to provide an accurate estimate of their location when compared with RClA. In Fig. \ref{fi:Multisource} we plot the Cramer-Rao bound corresponding to a single merged source at $\theta = 31.5^o$ in order to provide a lower bound for the performance of the algorithms. Since the Cramer-Rao bound for a single target should be lower than for multiple targets, the estimate variance is lower-bounded, but clearly not tightly so. RClA estimates the number of targets based on information that is actually used to provide the estimate of the source direction. Thus, if RClA is able to correctly detect the number of targets, it will provide a good estimate of their location as well. This is in contrast with root-Music which struggles to localize both targets outside of the high SNR regime.

\begin{figure}[t]
	\vspace{5ex}
	\begin{center}
		\includegraphics[width=\linewidth,trim=0 250 0 250]{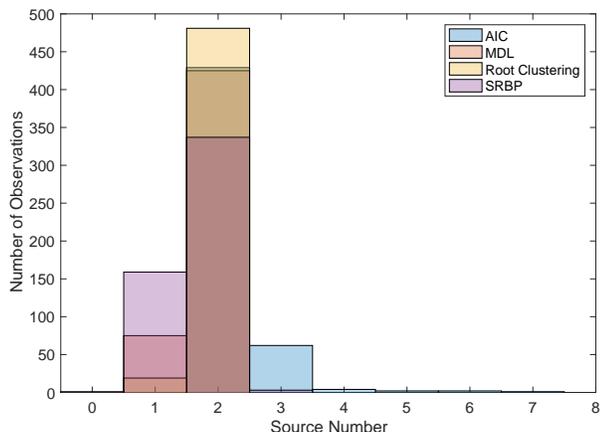}
	\end{center}
	\vspace{2ex}
	\caption{Histogram of number of source estimates for RClA, MDL, and AIC for closely located sources at SNR $= 3$dB.}
	\label{fi:hist}
\end{figure}

As a result, while root-MUSIC enjoys a performance advantage in scenarios with widely separated sources and in high SNR regions, RClA provides acceptable performance in all scenarios, and highly robust performance in the case of closely located sources in practical SNR regions.

In summation, the root-clustering algorithm and the root-certificate algorithm leverage two different concepts of perturbation from membership in a polynomial ideal. The first considers perturbations of roots, while the second considers perturbations of the coefficients. The simulations in these scenarios demonstrate the practical use of both. The root-clustering method offers superior performance detection to eigenvalue based methods. The root certificate algorithm supplies a notion of estimation optimality with respect to the observed system of polynomials.

\subsection{Convergence of Root-Certificate Algorithm}

We will first investigate the convergence rate of the Gauss-Newton iteration \eqref{eq:gauss_newton} in both the widely separated sources, and closely located sources scenario. Fig. \ref{fi:converge_wide} depicts the convergence rate of the Gauss-Newton iteration in the presence of widely separated sources, while Fig. \ref{fi:converge_close} depicts the convergence rate in the presence of closely located sources.

\begin{figure}[t]
	\vspace{4ex}
	\begin{center}
		\includegraphics[width=\linewidth,trim=0 250 0 250]{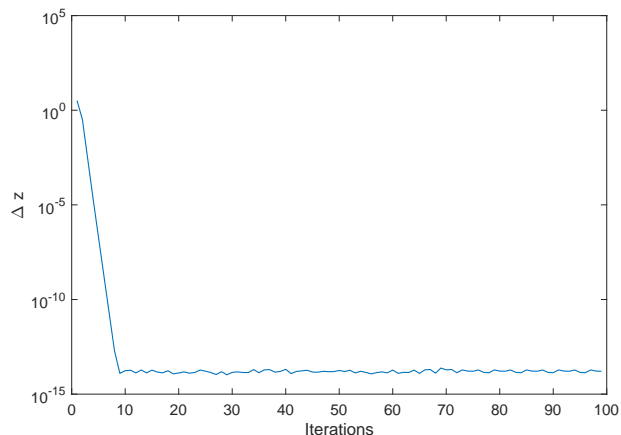}
	\end{center}
	\vspace{2ex}
	\caption{Convergence of the Root Certificate Algorithm. Two sources impinge on the array from directions $\Theta = [-10^o \ 20^o]$, respectively. Initial target estimates are provided by the Root-Clustering Algorithm.}
	\label{fi:converge_wide}
\end{figure}

\begin{figure}[t]
	\vspace{4ex}
	\begin{center}
		\includegraphics[width=\linewidth,trim=0 250 0 250]{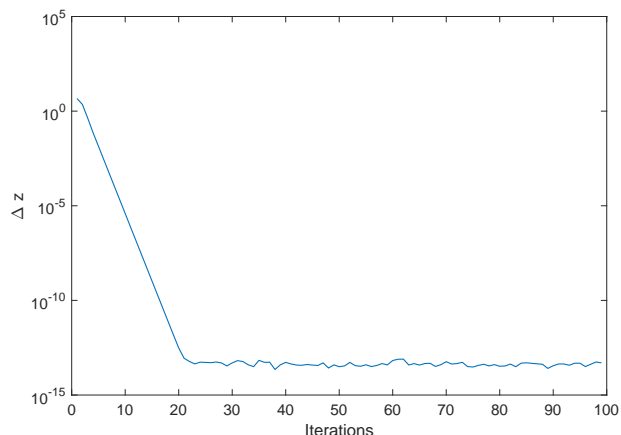}
	\end{center}
	\vspace{2ex}
	\caption{Convergence of the Root Certificate Algorithm. Two sources impinge on the array from directions $\Theta = [31^o \ 32^o]$, respectively. Initial estimates are provided by the Root-Clustering Algorithm.}
	\label{fi:converge_close}
\end{figure}

In Figs. 1 and 2, we observe that convergence of the Gauss-Newton iteration is linear. Fig. 1 demonstrates that, in the case of widely separated sources, the Gauss-Newton iteration converges to machine precision in less than $N$ iterations. In the case of closely located targets, we observe that convergence requires more iterations due to the fact that the system of polynomials is ill-conditioned. However, convergence remains linear.

In terms of detection algorithms, it should be stated that the SRBP has also high complexity. Separate eigendecompositions must be made $M$ times for each of $B$ bootstraps. In terms of $N$ however, $B$ can be on the order of $N^2$ (or larger), and $M$ can be on the order of $N$. Given that $M \cdot B$ eigendecompositions must be made, the total complexity in terms of $N$ becomes $\mathcal{O}(N^6)$, which is roughly equivalent to the root-certificate algorithm. The root-clustering algorithm, which we use for detection has complexity $\mathcal{O}(N^4)$. Moreover, the root-certificate algorithm is an estimation algorithm, whereas the SBRP is a detection algorithm.

\section{Conclusions and Discussion} \label{Concl}

A new criterion for the simultaneous detection and localization of an unknown number of sources has been introduced, and two algorithms based on this criterion were proposed. The proposed algorithms leverage the underlying algebraic structure present in the noise-eigenspace of the observation covariance matrix for both detection and estimation. The essential difference between our approach, and the dominant approach of source detection based on eigenvalues, and estimation based on eigenvectors is that we aim to estimate a single object, the maximal degree GCD, of a subset of the observed eigenvectors. The properties of this object are then the solutions to the separate problems. Specifically, the maximal degree is the number of targets, and the factors themselves are injectively related to the target locations. Notably, this implies that for every degree of this GCD, there is a factor which provides the estimate of the target location. In illustrative example 1, however, it was shown that this is not necessarily true of the eigenvalue-detection eigenvector-estimation paradigm. Targets could be estimable based on the observed eigenvectors, but not detectable on the basis of the eigenvalues alone. The first algorithm uses the root structure of the noise eigenvector. The second algorithm uses structured matrices to estimate the co-dimension of a solution manifold, and then estimate the nearest point on this manifold from the observed eigenvectors. The proposed algorithms have been compared to root-MUSIC in two different scenarios. In the first scenario, the number and locations of widely separated sources were estimated. In the second, the number and locations of closely located sources were estimated. Simulation results show that the first proposed algorithm offers significant performance benefits over information theoretic detection algorithms and localization using root-MUSIC in the case of closely located sources over a broad and practical range of SNRs.

The fundamental finding of the paper is that the noise subspace of the SCM, under the assumptions in the paper, is imbued with an algebraic structure which happens to be closed under any and all linear operations. This is not true of the signal subspace. More to the point, when we conduct an eigendecomposition of the SCM what are returned are linear invariants of the matrix, not the signal vectors themselves. The signal vectors themselves are all, by definition, points on the rational normal curve. While this is a powerful algebraic structure, it is not, in general, descriptive of the signal eigenvectors. However, it is rather trivial to show that any member of the noise subspace, invariant or not, has the algebraic structure which is leveraged by the analysis and algorithms contained in the paper. Thus, as long as we have a way of tracking the noise subspace, it is possible to use the algorithms contained in this paper.

One obvious limitation to this method is that it is suitable for point sources only. It is unclear how to adapt the algorithm to the estimation of spread sources, since univariate polynomial ideals, by definition have points as their corresponding varieties. As a limitation not necessarily of the algorithms in the paper, but of the analysis presented in the paper, we do not consider correlated sources, or of non-Gaussian noise distributions. Investigations of these scenarios are left for future work.

Another apparent limitation of the analysis contained in the main body of this work rely on the structure of a uniform linear array (ULA). However, if one allows for a pre-processing step, the analysis in this paper can be applied to arrays of arbitrary planar geometry. We showed that the noise eigenvectors of the sample covariance matrix exist in a principal polynomial ideal generated by the polynomial

\begin{align}
Q(x) = \prod_{l=1}^L (x - \alpha_l) 
\end{align}
where $\alpha_l$ are complex numbers whose phase arguments correspond to the target locations. However, from the definition of $\mathbf A$ in \eqref{eq:cov_matrix} it is clear that the polynomial structure of the noise subspace is a direct consequence of the uniform linear spatial sampling of the array.

In \cite{KoivunenBelloni} it was shown that the root-MUSIC algorithm can be applied to arbitrary arrays through a ``manifold separation'' technique. To perform the manifold separation technique, the Jacobi-Anger expansion is applied to the steering vectors of an arbitrary array to map them into those of a uniform array. The $n$-th element of the steering vector corresponding to an arbitrary planar array can be written as

\begin{align}
[\mathbf b(\theta)]_n &= e^{-j \omega_c \tau_n (\theta)} \label{eq:arb_array}
\end{align}
where $\omega_c \triangleq 2 \pi f_c$ is the angular frequency of the carrier waveform, and $\tau_n(\theta) \triangleq (-r_n/c)\rm{cos}(\gamma_n - \theta)$ is the propagation delay between the $n$-th antenna element and the centroid of the array, where $-r_n$ is the distance between the $n$-th antenna element and the array centroid, $c$ is the propagation velocity, and $\gamma_n$ is the polar position of the $n$-th antenna element.

Using the Jacobi-Anger expansion, \eqref{eq:arb_array} can be rewritten as 

\begin{align}
e^{j \kappa r_n \rm{cos}(\gamma_n - \theta)} &= \sum_{m = -\infty}^{\infty} j^m J_m(\kappa r_n) e^jm(\gamma_n - \theta) \nonumber \\
&= \frac{1}{\sqrt{2 \pi}}\sum_{m=-\infty}^{\infty} [\mathbf{G}(r_n,\gamma_n)]_{n,m} e^{-jm\theta} \label{eq:JA_vec}
\end{align}
where $[\mathbf{G}(r_n,\gamma_n)]_{n,m} \triangleq \sqrt{2 \pi} j^m J_m(\kappa r_n) e^{jm\gamma_n}$ and $J_m(\cdot)$ is the Bessel function of the first kind. From \eqref{eq:JA_vec} it is clear that the Jacobi-Anger expansion maps the non-polynomial steering vector corresponding to an arbitrary array in finite dimension, to a polynomial steering vector of infinite dimension. A finite approximation of $\mathbf b(\theta)$ is possible by arranging the coefficients $[G(r_n,\gamma_n)]_{n,m}$ in an $N \times M$ matrix $\mathbf{G}$. The finite approximation of the steering vector can thus be written as

\begin{align}
\mathbf b(\theta) \approx \mathbf G \mathbf h(\theta) \label{eq:arbarray_approx}
\end{align}
where $[\mathbf h(\theta)]_m \triangleq \frac{1}{\sqrt{2 \pi}} e^{jm \theta}$. The natural question of how large the sampling matrix $\mathbf G$ should be is both beyond the scope of this paper, and treated at some length in the reference \cite{KoivunenBelloni}.

With \eqref{eq:arbarray_approx} in mind, we can rewrite \eqref{eq:obsignal} and \eqref{eq:cov_matrix} as

\begin{align}
\mathbf{x}(t) &= \mathbf B \mathbf \mathbf s(t) + \mathbf{n}(t) \nonumber \\
&= \mathbf G \mathbf H \mathbf s(t) + \mathbf{n}(t)
\end{align}
\begin{align}
\mathbf R_{xx} &\approx \sigma_s^2 \mathbf{G H H}^H \mathbf G^H + \sigma_n^2 \mathbf I
\end{align}
respectively, where $\mathbf H$ has columns $\mathbf h(\theta)$ corresponding to each target.

Following the analysis of the previous section, the range of matrix product $\mathbf{GH}$ is contained in the span of the columns of the signal eigenvector matrix $\mathbf Q_s$ and is orthogonal to the noise eigenvector matrix $\mathbf Q_n$. One can then write the root-MUSIC polynomial for an arbitrary array as

\begin{align}
\mathbf h^H(\theta) \mathbf G^H \mathbf Q_n \mathbf Q_n^H \mathbf{Gh}(\theta)
\end{align}
after which the analysis in the main submission holds. Specifically, the columns of the matrix $\mathbf G^H \mathbf Q_n$ must lie in a principal polynomial ideal generated by \eqref{eq:generator}.

\appendix
We start by expressing $\boldsymbol{\epsilon}$ and finding its derivatives with respect to the real and imaginary parts of $\lambda_i$
\begin{align}
\boldsymbol{\epsilon} = \sum_{i=0}^{n-1} \lambda_i \lambda_i^* + &\sum_{i=0}^{n-1} \mu_i \mu_i^* + 2a \rm{Re} (\hat{\mathit{f}}(\alpha)) + 2b \rm{Im} (\hat{\mathit{f}}(\alpha))\\ &+ 2c \rm{Re} (\hat{\mathit{g}}(\alpha)) + 2d \rm{Im} (\hat{\mathit{g}}(\alpha)) \nonumber \\
\frac{\partial \boldsymbol \epsilon}{\partial \rm{Im} (\lambda_i)} &= 2\rm{Im}(\lambda_i) + 2\mathit{a}(-\rm{Im}(\alpha^i)) + 2\mathit{b}(\rm{Re}(\alpha^i)) \nonumber \\
\frac{\partial \boldsymbol \epsilon}{\partial \rm{Re}(\lambda_i)} &= 2\rm{Re}(\lambda_i) + 2\mathit{a}(Re(\alpha^i) - Im(\alpha^i))\\  &+ 2\mathit{b}(Im(\alpha^i)) \nonumber
\end{align}
Adding real and imaginary parts into a single equation at optimality, we write
\begin{align}
\lambda_i + \bigg( a \rm{Re}(\alpha^i) - &\mathit{b} \rm{Im}(\alpha^i) + j ( \mathit{b} \rm{Re}(\alpha^i) - \mathit{a} \rm{Im}(\alpha^i))\bigg) = 0 \nonumber \\
&\lambda_i + (a + jb)\alpha^{*i} = 0 \nonumber \\
&\lambda_i = -(a + jb)\alpha^{*i}\nonumber
\end{align}
Combining this equation with the definition of $\hat{\mathit{f}}(\alpha)$ and noting that $\hat{\mathit{f}}(\alpha) = 0$ one obtains
\begin{align}
\hat{\mathit{f}}(\alpha) = 0 &= \sum_{i=0}^{n-1} (f_i - (a + bj)\alpha^{*i})\alpha^i \nonumber \\
0 &= \sum_{i=0}^{n-1} f_i \alpha^i - (a + jb) \sum_{i=0}^{n-1} \alpha^{*i} \alpha^i \nonumber \\
a + jb &= \frac{f(\alpha)}{\sum_{i=0}^{n-1} \alpha^{*i}\alpha^i} \nonumber \\
&\implies \lambda_i = -\frac{f(\alpha)}{\sum_{i=0}^{n-1} \alpha^{*i}\alpha^i}\alpha^{*i} \nonumber
\end{align}
The derivation for $\mu_i$ is identical since all terms with $\lambda_i$ are constant with respect to $\mu_i$ and decoupled. The only thing that changes is the variable names. Noting this readily yields the equations
\begin{align}
c+jd = \frac{g(\alpha)}{\sum_{i=0}^{n-1} \alpha^{*i}\alpha^i},\         \mu_i = -\frac{g(\alpha)}{\sum_{i=0}^{n-1} \alpha^{*i}\alpha^i} \alpha^{*i} \nonumber
\end{align}
Substitution back into the equation for $\boldsymbol \epsilon$ yields the following formula for $\boldsymbol \epsilon_{\rm min}$
\begin{align}
\boldsymbol \epsilon_{\rm min} &= \frac{f(\alpha) f^*(\alpha) + g(\alpha) g^*(\alpha)}{\sum_{i=0}^{n-1}(\alpha^* \alpha)^i} \nonumber
\end{align}
To extend this derivation to an arbitrary set of polynomials, we rename the polynomials and the perturbation coefficients as $f_l(\alpha)$ and $\lambda_{i,l}$ respectively. Following the preceding derivation for $\frac{\partial \boldsymbol \epsilon}{\partial \lambda_{i,j}} = 0$ will yield the equations
\begin{align}
\lambda_{i,l} = -\frac{f_l(\alpha)}{\sum_{k=0}^{n-1}(\alpha^* \alpha)^k} \nonumber \\
\implies \boldsymbol \epsilon_{\rm min} = \frac{\sum_{l=1}^{L} f_l(\alpha) f_l^*(\alpha)}{\sum_{k=0}^{n-1}(\alpha^* \alpha)^k}. 
\end{align}

\bibliographystyle{IEEEbib}

\end{document}